\documentclass[english,twocolumn,prx,amssymb,amsmath,superscriptaddress]{revtex4}

\usepackage{amsfonts,amsthm,mathrsfs,eufrak,xspace,framed}
\usepackage{color}
\usepackage{algorithm}
\floatname{algorithm}{Protocol}
\usepackage{algorithmic}
\usepackage{subfigure}
\usepackage{amssymb}
\usepackage{amscd}
\usepackage{amsmath}    
\usepackage{multirow}
\usepackage{graphicx}
\usepackage{dcolumn}
\usepackage{bm}

\newtheorem{definition}{Definition}
\newtheorem{lemma}{Lemma}
\newtheorem{theorem}{Theorem}
\newtheorem{corollary}{Corollary}

\newcommand{\ket}[1]{\mbox{$\left| #1 \right\rangle$}}

\newcounter{protocol}

\newcommand{\be}{\begin{equation}}
\newcommand{\ee}{\end{equation}}
\newcommand{\bea}{\begin{eqnarray}}
\newcommand{\eea}{\end{eqnarray}}

\definecolor{mygreen}{rgb}{0,0.5,0}
\definecolor{myblue}{rgb}{0,0,0.75}
\definecolor{mymagenta}{cmyk}{0,1,0,0.12}

\usepackage{umoline}

\newcommand{\sect}[1]{~\\\noindent{\large{\bf{#1}}~\\~\\ \noindent}}
\newcommand{\subsect}[1]{\noindent{\bf{#1}}\\ \noindent}

\begin{document}
\title{Constant-round Multi-party Quantum Computation for Constant Parties}
\author{Zhu Cao}
\email{caozhu@ecust.edu.cn}
\affiliation{Key Laboratory of Advanced Control and Optimization for Chemical Processes of Ministry of Education, East China University of Science and Technology, Shanghai 200237, China}
\affiliation{Shanghai Institute of Intelligent Science and Technology, Tongji University, Shanghai 200092, China}

\maketitle

{\bf
One of the central themes in classical cryptography is multi-party computation, which performs joint computation on multiple participants' data while maintaining data privacy. The extension to the quantum regime was proposed in 2002, but despite two decades of research, the current state-of-the-art multi-party quantum computation protocol for a fixed number of parties (even 2 parties) requires unbounded rounds of communication, which greatly limit its practical usage. In this work, we propose the first constant-round  multi-party quantum computation protocol for a fixed number of parties against specious adversaries, thereby significantly reducing the required number of rounds. Our work constitutes a key step towards practical implementation of secure multi-party quantum computation, and opens the door for practitioners to be involved in this exciting field. The result of our work has wide implications to quantum fidelity evaluation, quantum machine learning, quantum voting, and distributed quantum information processing.
}

\sect{Introduction}
\label{sec:introduction}
The quantum technology is able to significantly boost the security level of network communication, forming a vibrant field of quantum cryptography \cite{Bennett:BB84:1984}. As the technology of pairwise quantum communication gradually reaches maturity both theoretically and experimentally \cite{LoChauQKD_99,ShorPreskill_00, Lo:MDIQKD:2012, arnon2018practical,liao2017satellite}, research on quantum networks with multiple nodes has gradually gained attention, and spurs a bunch of pioneering works on various quantum network tasks, including multi-node versions of quantum key distribution \cite{fu2015long}, quantum steering \cite{he2013genuine}, quantum teleportation \cite{yonezawa2004demonstration}, and dense coding \cite{jing2003experimental}. As it turns out, a quantum network with multiple nodes contains a much richer structure than its counterpart with two nodes \cite{gisin2020constraints}, presenting both new challenges and interesting new physics under its belt. Moreover, quantum network serves as a powerful infrastructure and has strong connections to quantum computation \cite{spiller2006quantum}, quantum metrology \cite{giovannetti2004quantum}, clock synchronization \cite{giovannetti2001quantum} and distributed quantum computing \cite{cirac1999distributed}. Advances in quantum networks can foreseeably facilitate fast progress in these related fields.

So far, the quantum cryptography community has been mainly focused on quantum key distribution \cite{Bennett:BB84:1984}, an intrinsically two-node primitive. It is natural to wonder whether quantum technology can find an equally important privacy-preserving application in the quantum network setting. To this end, it is helpful to draw inspiration from the classical cryptology literature.
In a classical network with multiple nodes, one of the central security goals is \emph{multi-party computation} (MPC), which is closely related to many other security topics, including fully homomorphic encryption \cite{gentry2009fully}, coin tossing \cite{blum1983coin}, oblivious transfer \cite{rabin2005exchange}, bit commitment \cite{naor1991bit},  user identification \cite{peacock2004typing}, authenticated key exchange \cite{bellare2000authenticated}, and zero-knowledge proofs \cite{feige1988zero}.  
As a motivation for MPC, consider the following scenario. Several employees want to vote on a controversial issue and only inform the manager whether a majority voted ``yes'' or not, keeping their individual opinions private. If there exists a trusted third-party, a potential solution is that the employees send their votes to the third-party and the third-party aggregates the votes and informs the manager the result, as shown in Fig.~\ref{fig:concept}(A). However, in reality, there is often no such third-party trustable by all employees. Multiparty computation (MPC) aims to accomplish the task in this challenging setting, by somehow exchanging encrypted messages among these untrusted employees, as shown in Fig.~\ref{fig:concept}(B). Mathematically, in MPC, $n$ parties aim to jointly compute a function $F$ on their private data $(x_1, x_2, \cdots, x_n)$, while not revealing their private data except which is absolutely needed to compute $F$. In other words, even if $n-1$ parties are dishonest and collude, they cannot learn the honest party's input more than which can be inferred from their own inputs and the output $z=F(x_1, x_2, \cdots, x_n)$.

\begin{figure}[htb]
\centering \includegraphics[width=8cm]{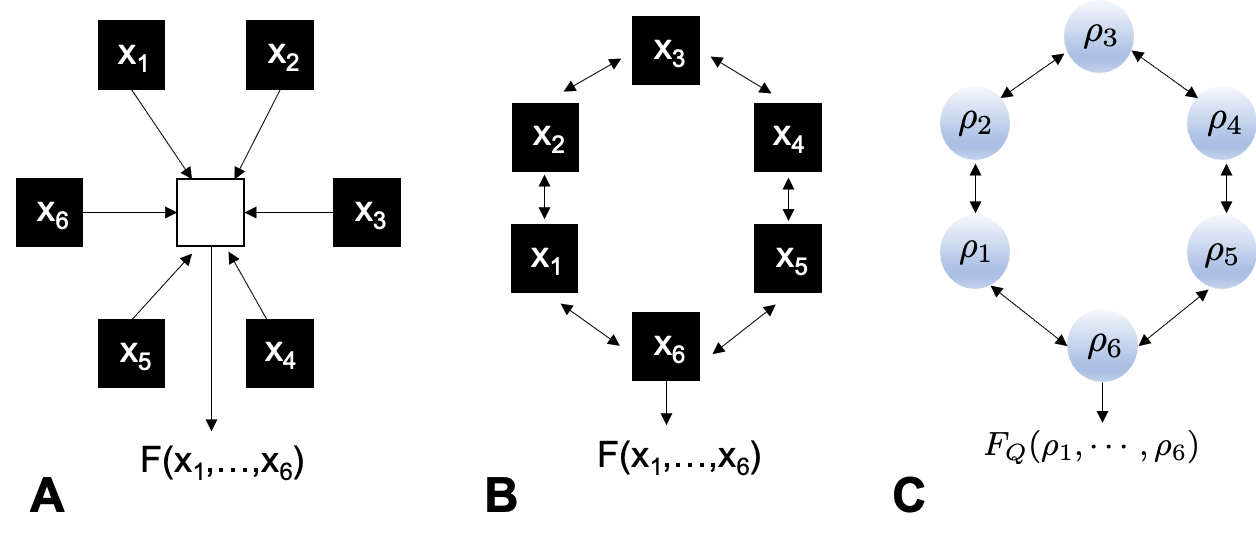}
\caption{{\bf Problem setting. (A)} Several untrusted parties (black squares) send their private classical data to a trusted central node (white square), which subsequently outputs the result of a classical circuit $F$ on the data. {\bf (B)} Without a trusted node, these untrusted parties performs MPC to get the result, maintaining the property that their private data are hided from each other. {\bf (C)} Upgraded with private quantum data (circles), the parties perform MPQC to get the result of a quantum circuit $F_Q$.}
\label{fig:concept}
\end{figure}

Multi-party computation is widely applicable to a vast number of situations where participants have sensitive information but do not trust each other, such as electronic transactions \cite{chellappa2002perceived}, auctions \cite{riley1981optimal}, and contract signing \cite{garay1999abuse}. The notion of MPC is first initiated by Yao \cite{yao1986generate}, who also proposed a two-party MPC protocol.  Later, MPC is extended to multiple parties \cite{micali1987play}, which has round complexity linear in the depth of the circuit that computes $F$. The round number is reduced to a constant in \cite{beaver1990round,rogaway1991round}, which takes the MPC protocol in \cite{micali1987play} as a subroutine. Recently, the round complexity for semi-honest adversaries is further reduced to two with the minimal assumption that a two-round oblivious transfer (OT) exists \cite{garg2018two,benhamouda2018k}. It has also been shown that MPC with one round is impossible \cite{cohen2020broadcast}, hence two rounds are both necessary and sufficient for MPC against semi-honest adversaries. For malicious adversaries, it has been shown also recently that four rounds are both sufficient  \cite{badrinarayanan2018promise,halevi2018round}  and necessary \cite{applebaum2020round} for secure MPC.

In a quantum network, by direct analogy, it is conceivable that multi-party quantum computation (MPQC) also plays a paramount role in quantum network security, and has wide applications to many quantum network tasks. Secure multi-party computation is first generalized to the quantum regime by Claude et al. \cite{crepeau2002secure}. The parties now hold quantum data $\rho_i$ instead of classical data $x_i$. A pictorially illustration of MPQC is shown in Fig.~\ref{fig:concept}(C). Currently, both the best two-party MPQC \cite{dupuis2010secure} and  the best multi-party MPQC \cite{dulek2020secure}  in terms of round complexity has round number linear in the quantum circuit depth $d$, and hence are unbounded as the circuit depth grows. This is in stark contrast with Yao's original secure two-party computation protocol, where only a constant number of rounds is needed. Following Yao's paradigm, we design a constant-round two-party MPQC and a constant-round multi-party MPQC for a fixed number of parties against specious adversaries, significantly reducing the round requirement of MPQC.

Technically, our work exploits a tool called decomposable quantum random encoding (DQRE), which is a quantum analog of Yao's garbled circuit in his constant-round two-party protocol construction. DQRE encrypts a  quantum input $\rho$ and a quantum function $F_Q$ so that only the value of the function on this input $F_Q(\rho)$ can be obtained from the encryption, but not the input $\rho$ or the function  $F_Q$ themselves. Our two-party MPQC protocol takes the component DQRE as a black box, while our multi-party MPQC protocol uses DQRE in a non-black-box way. For the multi-party MPQC protocol, we also develop a technique, called ``qubit flipping'', which maybe of independent interest.
 Since our work significantly reduces the round requirement of MPQC from an unbounded  number to a constant, our work constitutes a significant step forward towards practical implementation of MPQC. As an application, our result also significantly cuts resource requirement in distributed quantum information tasks, including quantum fidelity evaluation \cite{vanivcek2003semiclassical}, quantum machine learning \cite{biamonte2017quantum}, and quantum voting \cite{vaccaro2007quantum}, as these tasks can be instantiated as a MPQC problem.
 
 


\sect{Results}
Before presenting our protocols, let us
first give a formal definition of MPQC against specious adversaries. To begin with, let us give the definition of a specious adversary.
\begin{definition}[Specious adversary \cite{dupuis2010secure}]
An adversary in a protocol is called specious if at every step of the protocol, it can transform its actual state to 
one that is indistinguishable from the ideal state. 
\end{definition}
From its definition, it can be seen that a specious adversary is a quantum analogue of a semi-honest adversary, which follows the protocol but is curious about honest parties' inputs. Compared with an honest party, a specious adversary can in addition utilize an arbitrarily large quantum
memory to store extra information, and eliminates this memory when it is called to be compared with the state if the adversary were honest.

We are now ready to give a definition for multi-party quantum computation against a specious adversary.
\begin{definition}[MPQC against specious adversary]
A multi-party quantum computation protocol for a quantum operation $F$ on the quantum inputs $x_1, x_2, \dots, x_n$ of party 1, party 2, $\dots$, party $n$ respectively against a specious adversary satisfies the following properties:
\begin{enumerate}
\item At the end of the protocol, every party gets the result $F(x_1, x_2, \dots, x_n)$. 
\item $\epsilon$-privacy: Throughout the protocol, every party is ignorant of all information except the final result $F(x_1, x_2, \dots, x_n)$ and its own input $x_i$, 
i.e., for any specious party $A$ holding the input $x_i$, there is a simulator $S$ that only takes $x_i$ and $F(x_1, x_2, \dots, x_n)$ as inputs and simulates the view of $A$ at every step of the protocol. That is, at every step of the protocol, for any distinguisher $D$, the advantage that $D$ can distinguish the view of $A$ and the output of $S$ is at most $\epsilon$.
\end{enumerate}
\end{definition}

\subsect{MPQC for Two Parties}
\begin{figure*}[htb]
\centering \includegraphics[width=15cm]{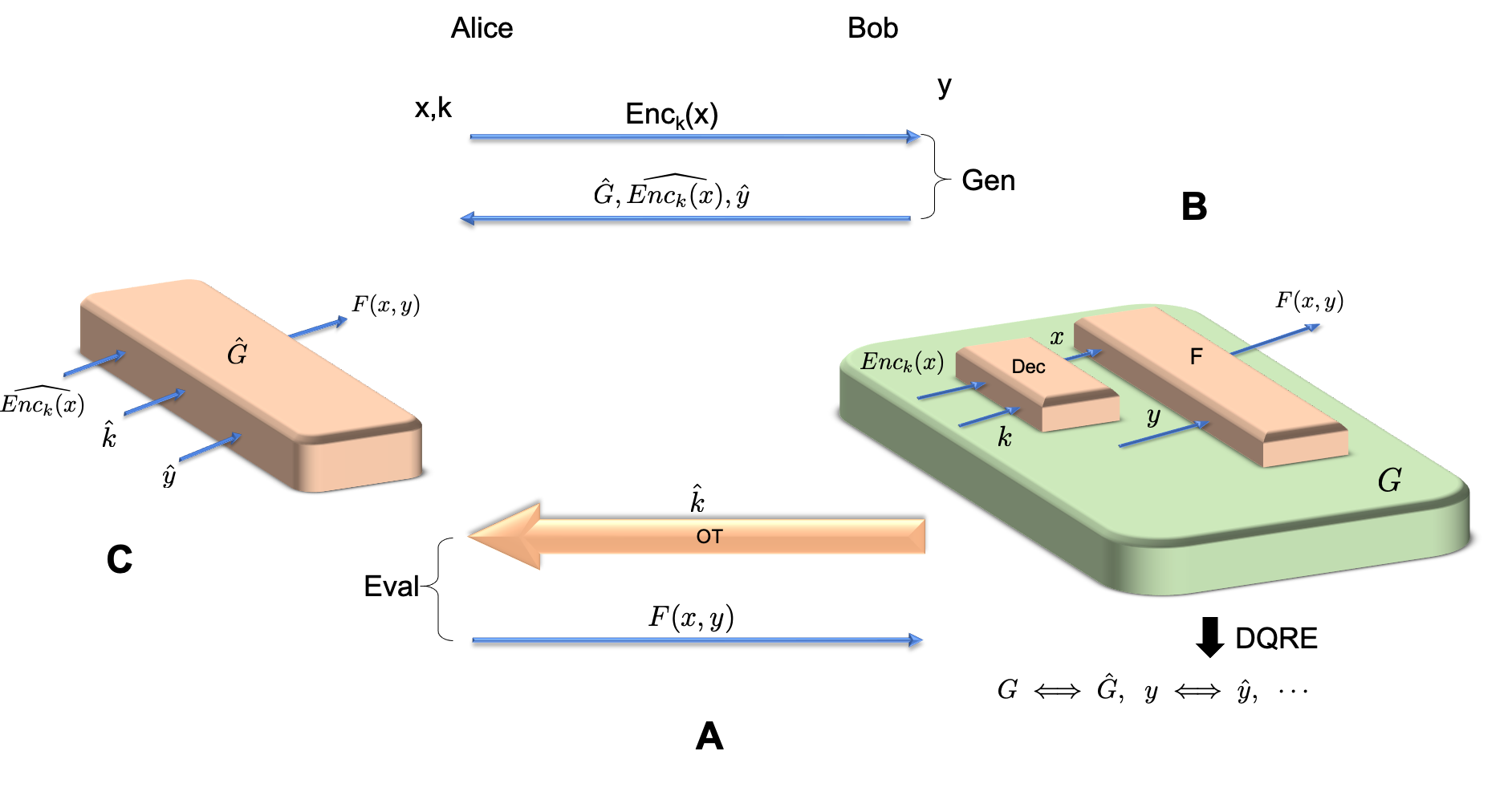}
\caption{ {\bf MPQC protocol for two parties. (A)} Alice holds a quantum input $x$ and a classical random string $k$. Bob holds a quantum input $y$. In the first round, Alice sends $Enc_k(x)$ to Bob. Bob generates a garbled program for a circuit $G$, and sends $\hat{G}, \widehat{Enc_k(x)}, \hat{y}$ to Alice. Alice interacts with Bob using an OT protocol to get $\hat{k}$. Alice then runs the evaluation circuit to get $F(x,y)$. Finally, Alice sends $F(x,y)$ to Bob, so that both parties get the computation result. 
 {\bf (B)} Illustration of the $G$ circuit. It first takes $Enc_k(x)$ and $x$ as inputs and uses a decryption algorithm $Dec$ to get $x$. Then it takes $y$ as the second input, feeds them to the circuit $F$, and gets the output $F(x,y)$ {\bf (C)} Illustration of the evaluation circuit. According to the property of the garbled program, the garbled circuit $\hat{G}$ takes $\widehat{Enc_k(x)}, \hat{k}, \hat{x}$ as inputs and outputs $F(x,y)$.} 
\label{fig:twoparty}
\end{figure*}
In this section, we present the protocol of MPQC for two parties. 
An overview of the protocol is shown in Fig.~\ref{fig:twoparty}.
As it can be seen, the protocol relies on two cryptographic primitives, decomposable quantum random encoding (DQRE) and oblivious transfer (OT). Hence, we first recall their definitions:
\begin{definition}[DQRE]
A decomposable quantum random encoding of a quantum operation $F$ and a quantum state $x$ is a quantum state $\hat{F}(x)$ that satisfies the following three properties:
\begin{enumerate}
\item F(x) can be decoded from $\hat{F}(x)$.
\item $\epsilon$-privacy: $\hat{F}(x)$ reveals almost no information of $F$ and $x$ except $F(x)$, i.e., there is a simulator $Sim$ such that for any distinguisher $D$ and any side information $y$ (e.g., $y$ can be $x$ or $F$), the advantage that $D$ can distinguish $(\hat{F}(x),y)$ from $(Sim(F(x)),y)$ is upper bounded by $\epsilon$.
\item $\hat{F}(x)$ encodes each qubit of $x$ independently.
\end{enumerate}
\end{definition}

\begin{definition}[OT against specious adversaries]
An oblivious transfer is a two-party protocol in the following setting. 
The two parties are called Alice and Bob. Alice holds a bit $b$ unknown to Bob. Bob holds two quantities $y_0$ and $y_1$ that are initially unknown to Alice. An $\epsilon$-secure oblivious transfer against specious adversaries satisfies the following properties:
\begin{enumerate}
\item At the end of the protocol, Alice knows $y_b$.
\item At the end of the protocol, Alice knows nothing about $y_{1-b}$. More precisely, if Alice is the specious adversary, then there is a simulator $S_1$ that only takes $b$ and $y_b$ as inputs such that for any distinguisher, its advantage to distinguish Alice's view and the output of the simulator is bounded above by $\epsilon$.
\item At the end of the protocol, Bob knows nothing about $b$. More precisely, if Bob is the specious adversary, then there is a simulator $S_2$ that only takes $y_0$ and $y_1$ as inputs such that for any distinguisher, its advantage to distinguish Bob's view and the output of the simulator is bounded above by $\epsilon$.
\end{enumerate}
\end{definition}

Given a quantum-secure public-key encryption scheme, an OT protocol that is $\epsilon$-secure against a specious adversary exists (see Methods). In addition, DQRE exists given a quantum-secure pseudorandom generator:
\begin{theorem}[Computational DQRE \cite{brakerski2020quantum}]
Let $\lambda$ denote the security parameter.
Assume the existence of pseudorandom generator against quantum adversary, there exists a DQRE scheme that 
has the following properties:
\begin{itemize}
\item The encoding can be computed by a $QNC_f^0$ circuit, which is a circuit of constant depth with bounded-arity gates and unbounded-arity quantum fan-out gates. 
A quantum fan-out gate performs the function $\ket{x}\ket{y_1}\cdots\ket{y_n}\rightarrow\ket{x}\ket{y_1\oplus x}\cdots\ket{y_n\oplus x}$.
The decoding can be computed in polynomial time in $\lambda$ and the circuit size $s$.
\item For any polynomial $q$, there exists a negligible function $\epsilon$ such that the scheme is $\epsilon(\lambda)$-private for any $q(\lambda)$-size circuit.
\end{itemize}
\end{theorem}

Based on these two primitives, Protocol~\ref{Fig:2PQC} presents the scheme of MPQC for two parties, as illustrated in Fig.~\ref{fig:twoparty}(A).
Here, the information $y_{1-b}$ is not given to Alice so that she will not know the value
$F(\cdots, 1-b, \cdots)$, which is a quantity that cannot be always obtained from $F(\cdots, b, \cdots)$ and $b$.
The security of the protocol is given in Theorem \ref{thm:twoparty}.

\begin{algorithm}
\caption{\textsc{MPQC for two parties}}
\begin{flushleft}
\textit{Input}: The two parties are called Alice and Bob, who hold $x$ and $y$ respectively.  The value they aim to compute is $F(x, y)$.
\end{flushleft}
\begin{algorithmic}[1]
\STATE
Alice sends her input qubits encrypted by quantum one-time pad (QOTP) \cite{ambainis2000private} to Bob. That is, every Alice's input qubit is applied $I, \sigma_x, \sigma_y, \sigma_z$ randomly.
\vspace{0.2cm}
\STATE
Bob generates a DQRE, which includes a garbled circuit of $G$, labels of Alice's encrypted input and QOTP key, and labels of Bob's own input. The circuit $G$ first decrypts Alice's input with the QOTP key, and evaluates the function $F$ with Alice and Bob's inputs, as illustrated in Fig.~\ref{fig:twoparty}(B). 
\vspace{0.2cm}
\STATE
Bob sends the part of DQRE that he can compute to Alice, including  $\hat{G},\widehat{Enc_k(x)},\hat{y}$. Bob in addition sends the label-value correspondence for output wires.
\vspace{0.2cm}
\STATE 
Alice and Bob perform a classical OT for each bit of Alice's QOTP key $k$, 
Alice chooses a value $b$ in $\{0,1\}$ and selects the label $y_b$ from Bob such that
$b$ is not known to Bob and Alice is ignorant of other labels, namely $y_j$ for $j\not = b$.
Alice then gets the active label $\hat{k}$ of the QOTP key $k$ without knowing the inactive labels.
\vspace{0.2cm}
\STATE Alice calculate $\hat{G}(\widehat{Enc_k(x)},\hat{y},\hat{k})$ to get $F(x,y)$, as illustrated in Fig.~\ref{fig:twoparty}(C).
\vspace{0.2cm}
\STATE Alice sends the result $F(x,y)$ to Bob.
\end{algorithmic}
\label{Fig:2PQC}
\end{algorithm}

\begin{theorem}
\label{thm:twoparty}
Assuming the existence of a quantum-secure public-key encryption scheme and a quantum-secure pseudorandom number generator, Protocol 1 achieves constant-round two-party quantum computation.
\end{theorem}

\begin{proof}[Proof sketch]
Let us consider Bob first. The only inputs that Bob receives from Alice 
is Alice's encrypted inputs in the first round, and the messages
from Alice during the OT protocol.  By the definition of OT, Bob 
knows nothing about Alice during OT. In addition,  it can be shown that 
the encrypted inputs from Alice in the first round also contains no information (see Supplementary Materials).
Bob's simulator is as follows. In the first step, the simulator outputs random quantum strings.
In the second step, the simulator just outputs what the simulator $S_2$ in OT outputs. 
The case of Alice is more complex and is shown in Supplementary Materials.

\end{proof}

\subsect{MPQC for Multiple Parties}
\begin{figure*}[htb]
\centering \includegraphics[width=15cm]{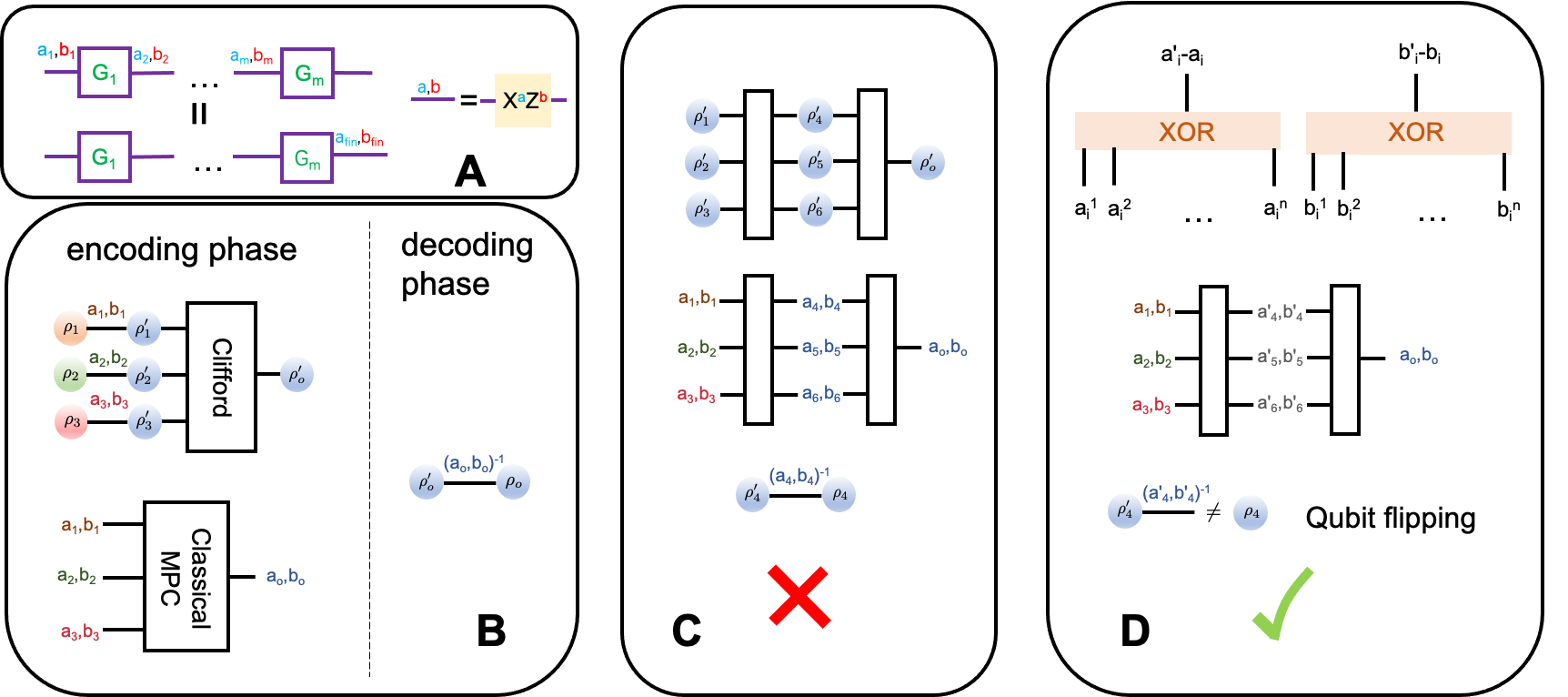}
\caption{{\bf MPQC for multiple parties and a Clifford circuit. (A)} Property of Clifford circuits. If the gates $G_1,\cdots, G_m$ are Clifford, the operations $X^{a_i} Z^{b_i} $ can be deferred to the end of the circuit. {\bf (B)} By the property of a Clifford circuit, its garbled program can consist of a quantum state $\rho'_o$ and a classical correction $a_o,b_o$. Here $\rho'_o$ is obtained by a two-step process. In the first step, party $i$ uses a QOTP key $a_i, b_i$ to encrypts his quantum input $\rho_i$ to $\rho_i'$ and sends $\rho_i'$ to party 1. In the second step, party 1 puts $\rho'_1,\cdots, \rho'_n$ to the Clifford circuit and calculates $\rho'_o$. The correction $a_o, b_o$ can be calculated through a classical MPC with inputs $a_i, b_i, 1\le i \le n$. In the evaluation, party 1 simply applies the correction $a_o, b_o$ on the quantum state $\rho'_o$ to obtain $\rho_o$. {\bf (C)} Suppose the Clifford circuit consists of two layers, and the quantum states between the layers if measured are $\rho'_4, \rho'_5, \rho'_6$. If the correction $a_4, b_4$ is obtained by the adversary, then the value $\rho_4$, which is the value between the layers when the Clifford circuit takes $\rho_1, \rho_2, \rho_3$ as inputs, is revealed by correcting  $\rho'_4$ with $a_4, b_4$. This is undesirable as it gives the adversary extra information on the users' private inputs. {\bf (D)} To resolve the previous issue, we apply a ``qubit flipping'' operation to each intermediate quantum state $\rho_i$ for each intermediate wire $\omega_i$. The parameters of qubit flipping $a'_i-a_i$ and $b'_i - b_i$ are obtained by XORing each party's local random bit $a_i^j$ and $b_i^j$ in a MPC way. After qubit flipping is applied, $a_i, b_i$ equivalently becomes $a_i', b_i'$ and can no longer retrieve any intermediate quantum value $\rho_i$.} 
\label{fig:clifford}
\end{figure*}
Next, we turn to the multi-party case. We note that if we use a simple extension of the two-party protocol
for multiple parties, namely one of the parties $A$ generates the garbled program and one of the parties  $B$ acts as the evaluator
to evaluate the outcome, then the privacy is broken. Indeed, if $A$ and $B$ collude, by the fact that $A$ knows all the correspondence between the wire values and the wire labels, and the fact that $B$ knows all wire labels, they can recover the input values of all parties!
To overcome this attack, instead of generating the garbled circuit 
by a single party, all parties should participate in the generation of the garbled circuit. 

Hence, we need to use a MPQC to generate the garbled circuit, but the round number of this MPQC does not need to be  a constant.  For example, we can utilize the following MPQC construction which has a round number linear in the circuit depth.
\begin{theorem}[MPQC with round number linear w.r.t. the circuit depth \cite{dulek2020secure}]
Assume the existence of a classical MPC secure against quantum adversaries, there exists a MPQC secure against quantum adversaries that has a round number $O(nd)$, where $n$ is the number of parties and $d$ is the circuit depth.
\end{theorem}

Combined with the result from Ref.~\cite{agarwal2020post}, we have the following theorem.
\begin{theorem}[MPQC with weaker assumptions \cite{dulek2020secure,agarwal2020post}]
\label{thm:multipre}
Assuming super-polynomial quantum hardness of LWE and quantum AFS-spooky encryption \footnote{We refer the readers to Ref.~\cite{agarwal2020post} for the definitions of LWE and quantum AFS-spooky encryption.}, there exists a MPQC secure against quantum adversaries that has a round number $O(nd)$, where $n$ is the number of parties and $d$ is the circuit depth.
\end{theorem}

Now we are ready to present our MPQC protocol in the multi-party case. Its high-level description is shown in Protocol~\ref{alg:MPQC}.
For the distributed computation of the garbled program, let us consider a Clifford circuit first. The construction utilizes the following crucial property of Clifford circuits. For single-qubit Clifford gates $G_1, G_2, \cdots, G_m$, we have that for arbitrary $a_i,b_i \in \{0,1\}, 1\le i \le m$, there exist $a_{fin}, b_{fin}$ such that
\begin{equation}
G_m X^{a_m}Z^{b_m} \cdots   G_1 X^{a_1}Z^{b_1} = X^{a_{fin}}Z^{b_{fin}}G_m \cdots G_1,
\end{equation}
which is illustrated in Fig.~\ref{fig:clifford}(A). In other words, the Pauli operations can be deferred to the end of the circuit after all other Clifford operations are performed. 
The same holds for multi-qubit Clifford gates, with single-qubit Pauli operations replaced by tensor products of Pauli operations.

With this Clifford property, we can design the following MPQC for a Clifford circuit. Suppose the $n$ parties holds $n$ states $\rho_1, \cdots, \rho_n$ respectively. Each state $\rho_i$ is first transformed by a QOTP key $(a_i,b_i)$, namely
\begin{equation}
\rho_i' = X^{a_i}Z^{b_i}\rho_i.
\end{equation}
The resulting states are sent to one of the parties, which feeds $\rho_1', \cdots, \rho_n'$ to the Clifford circuit and obtains $\rho_o'$. Only one round of communication is needed for this step. According to the Clifford property, this masked output $\rho_o'$ differs from the true output $\rho_o$ only by a Pauli operation $X^{a_{o}} Z^{b_{o}}$.
Therefore, we let the final output of the garbled program be a quantum state $\rho_o'$ and a correction $a_{o},b_{o}$ to $\rho_o'$. Here, $\rho_o'$ is random without knowing the values of $a_{o}, b_{o}$, and the values $a_{o}, b_{o}$ can be computed through a classical MPC using the values $(a_i,b_i), 1\le i\le n$.  The scheme is illustrated in Fig.~\ref{fig:clifford}(B).

\begin{algorithm}
\caption{\textsc{MPQC for multiple parties}}
\begin{flushleft}
\textit{Input}: $n$ parties hold quantum inputs $x_1,x_2,\cdots, x_n$ respectively.  The value to compute is $F(x_1,x_2,\cdots, x_n)$. 
\end{flushleft}
\begin{algorithmic}[1]
\STATE
Each party first encrypts his/her input by a QOTP,  and sends one copy to all other parties.
\vspace{0.2cm}
\STATE
These parties generate a garbled circuit and associated wire labels for the following function $H$ 
in a distributed fashion, using a nonconstant-round MPQC protocol such as the one in Theorem \ref{thm:multipre}. The function $H$ first decrypts each party's encrypted input using the corresponding QOTP key,
and then performs the quantum operation $F$ on the inputs.
In the end, each party gets the garbled circuit and the  input labels.
\vspace{0.2cm}
\STATE
Given the   input labels and the garbled gates, each party evaluates on its own following the topology of the circuit and obtains 
$F(x_1,x_2,\cdots, x_n)$.
\end{algorithmic}
\label{alg:MPQC}
\end{algorithm}

\begin{figure*}[htb]
\centering \includegraphics[width=17cm]{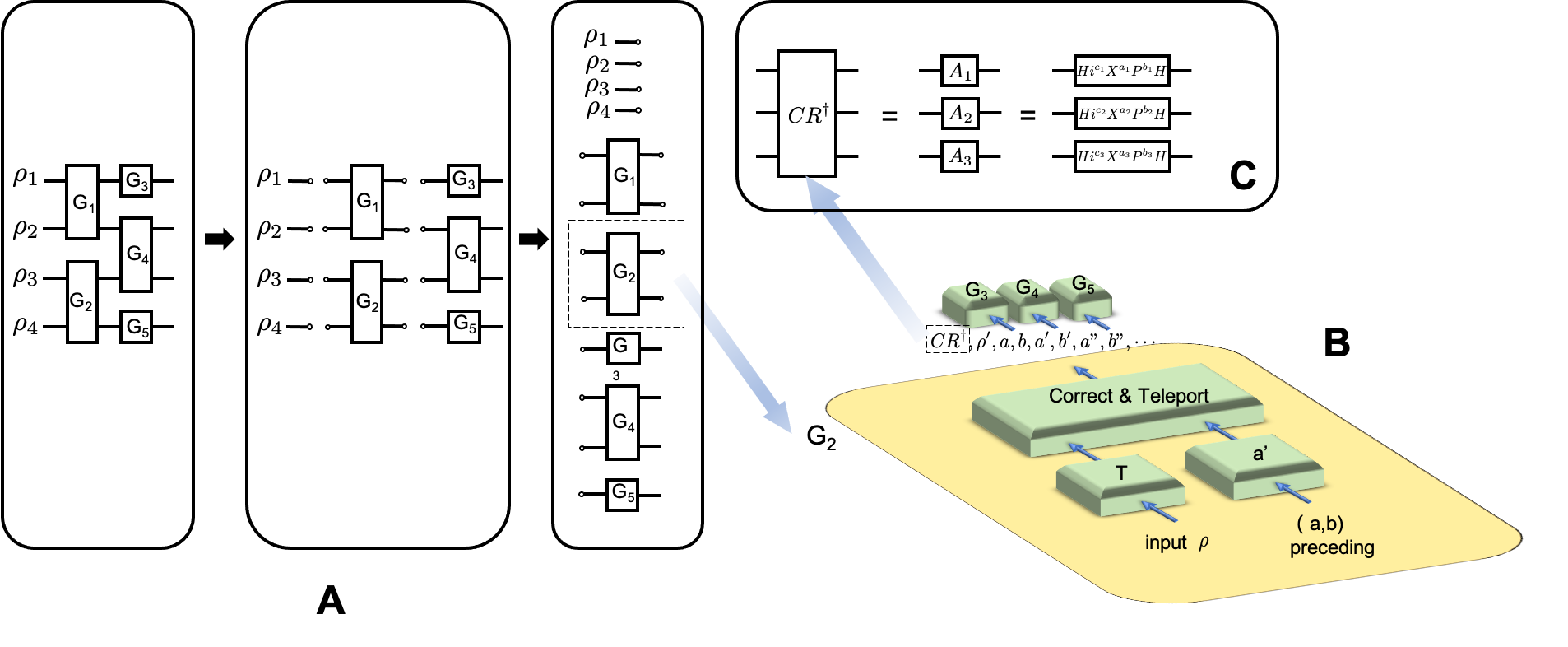}
\caption{{\bf MPQC for multiple parties and a general circuit. (A)} Compiler of a general circuit. The output of a gate is teleported to the input of another gate. Each input is viewed as one of the outputs of an initial gate $G_0$. Then the encryption of a circuit of arbitrary depth can be compressed to depth 1 by deferring the teleportation process. 
Note that the decryption process, which requires no communication between the parties, still proceeds in a sequential way. 
{\bf (B)} The detailed procedure for the encryption of a gate $G$. It takes a quantum state $\rho$ and labels of the correction values of preceding gates $(a,b)'s$ as inputs. It applies $G$ directly on the quantum input $\rho$ to obtain $\rho'$. Then a correction-and-teleport gadget takes $\rho'$ and $(a,b)'s$ as inputs, applies a correction on $\rho'$ according to $(a,b)'s$, and then teleports the corrected quantum state. Since the correction can only be obtained at the decryption stage, it is in the form of $C R^\dagger$. The encryption of $G$ also outputs one copy of its own teleportation correction $a,b$ for each subsequent gate. {\bf (C)} The quantum operation $CR^\dagger$ can be decomposed into a tensor product of single qubit gates. Each single qubit gate can be further represented as $H i^{c_i} X^{a_i} P^{b_i} H$. Here, $a_i,b_i,c_i$ are masked by qubit flipping similar to the Clifford case.} 
\label{fig:compiler}
\end{figure*}

Not all parts of this construction survives to a general circuit.
However, one of the parts still plays an important role in general circuits. We call this part ``\emph{qubit flipping}'' technique. To understand 
this technique, let us consider a depth-2 Clifford circuit and expand the MPQC construction for the Clifford circuit in this scenario. 
By the expansion, the classical MPC will also contain two layers. As shown in Fig.~\ref{fig:clifford}(C), let $\rho_4', \rho_5', \rho_6'$ be the quantum state between the two layers of the Clifford circuit and 
let $(a_4, b_4), (a_5, b_5), (a_6, b_6)$ be the correction values between the two layers in the classical MPC. In an oversimplified implementation of the classical MPC, acting $(X^{a_4}Z^{ b_4})^{-1}$ on $\rho_4'$ recovers $\rho_4$, which is the first qubit between the two layers if the Clifford circuit were inputted $\rho_1, \rho_2, \rho_3$. However, $\rho_4$ reveals extra information about the original quantum inputs other than the final output $\rho_o$, hence this implementation fails. 

To resolve this problem, in classical MPC, a flip bit $a_i'-a_i$ is added to the value $a_i$ for  every intermediate wire $i$ (see Methods for a full description of classical MPC), as shown in Fig.~\ref{fig:clifford}(D). The flip bit $a_i'-a_i$ is obtained by XORing the shares of the flip bit $a_i^j$ from all parties. A similar flip bit is added to $b_i$. Then $\rho_4$ can no longer be obtained from $\rho_4'$ and $(a_4',b_4')$. Equivalently, we can view the $n$ parties performed a ``qubit flipping'' operation  $X^{a_i-a_i'} Z^{b_i-b_i'}$  on each intermediate-wire quantum state $X^{a_i'}Z^{b_i'} \rho_i$ to hide the semantic value of this non-output quantum wire. 
We will utilize this technique again in the protocol for general circuits. 


Now let us consider a general circuit that can perform universal quantum computation.
Note that Clifford gates alone are insufficient to achieve universal quantum computation. To achieve universal quantum computation, $T$ gates are additionally needed.
However, $T$ gates make the garbled program more complex, as $T X^a Z^b = X^{a'}Z^{b'} T$ does not hold for all  $(a,b)$.
Hence, we can no longer put the encrypted inputs $\rho_1', \cdots, \rho_n'$ into the circuit, obtain the output, and later perform Pauli corrections.
In order to maintain constant rounds of communication, the critical idea here is to decouple multiple gates in the circuits through the use of EPR pairs and quantum teleportation.
As seen in Fig.~\ref{fig:compiler}(A), firstly, for each pair of gates $G_i,G_j$ where an output qubit of $G_i$ is the input qubit of $G_j$, we teleport the output of $G_i$ to the input of $G_j$ using an EPR pair. The inputs $\rho_1,\cdots,\rho_n$  are viewed as outputs of a  virtual gate $G_0$ and are handled similarly to pair of gates. Then, we can compress the circuit into a depth one circuit.

For each gate $G_i$, it takes half of the EPR pair $\rho$ as input together with the teleportation corrections $a,b$ of the preceding gates. Consider the most complicated case $G_i =T$, as illustrated in Fig.~\ref{fig:compiler}(B). The input $\rho$ goes through the gate $T$ first and becomes $\rho_T$. Next the quantum state goes through a deferred correction due to previous teleportation, and is then teleported to the next gate. For the correction, teleportation corrections $a,b$ of the preceding gates are required in addition to the state $\rho_T$. Here, the correction-and-teleport Clifford circuit $C$ is implemented through a group-randomizing DQRE. 

The purpose of this group-randomizing DQRE is to hide the semantic quantum values of the wires. It consists of a tensor product  of single qubit randomizers, denoted by $R$, on the state $\rho_T$ so that $\rho'=R(\rho_T)$, where each single qubit randomizer comes from the PX group. It also consists of a classical description of $CR^\dagger$ (note that $CR^\dagger(\rho') = C(\rho_T)$). In short, the output of the DQRE for the correction-and-teleport circuit $C$ mainly consists of $\rho'$ and $CR^\dagger$. In addition, the output also consists of the correction value for the teleportation of $G_i$ itself. We provide this correction value for successor gates to $G_i$.

We now adapt this group-randomizing DQRE for a correction-and-teleport Clifford circuit so that it is not generated by Bob alone, but jointly by all $n$ parties. Since each PX group element $i^c X^a P^b$ can be determined by three classical numbers $a\in \{0,1\}$ and $b,c\in \{0,1,2,3\}$, we let the $n$ parties perform five joint XOR operations to determine these five classical bits for each qubit, and then apply the corresponding $R$ on the quantum state to obtain $\rho'$ during gate encryption. Crucially, the encryption process is still of constant quantum circuit depth.
The decoding classical operation $CR^\dagger$ is adjusted similar to $R$, with a slight difference that it is a tensor product of conjugated PX group elements, which are of the form $Hi^c X^a P^bH$. 
 An illustration is shown in Fig.~\ref{fig:compiler}(C).


The security of this multi-party protocol is given in the following theorem:
\begin{theorem}
\label{thm:multiparty}
Protocol 2 compiles a MPQC scheme of which the round number is circuit-depth dependent to a MPQC scheme with round number independent of the circuit depth.
\end{theorem}

\begin{proof}[Proof sketch]
For a general quantum circuit, the parties mask each quantum wire (including input wire) through an element from the  PX group. The masking proceeds in two steps. In the first step, the $n$ parties XOR their values to determine five classical bits. This is a classical MPC and takes constant rounds. In the second step, the $n$ parties use these values to group randomize the wires according to $i^c X^a P^b$. This is a constant-depth quantum circuit, and hence by the assumption of the theorem, can be computed in constant rounds. 
By the property of DQRE, the rest part of the garbled program can also be computed in constant depth and hence can be computed in constant rounds by a MPQC scheme of which the round number is circuit-depth dependent. This finishes the part of proof for round complexity.
The details for the proof of security are deferred to Supplementary Materials.
\end{proof}

By Theorem \ref{thm:multiparty}, we have the following corollary:
\begin{corollary}
\label{cor:multiparty}
Assuming the existence of a secure multi-party quantum computation scheme with round number only as a function of the circuit depth, Protocol 2 is a constant-round multi-party quantum computation.
\end{corollary}

\sect{Discussion}
\label{sec:discuss}
We have shown that constant-round two-party quantum computation and constant-round multi-party quantum computation are possible under mild assumptions by designing the first protocols that achieve these goals. We have also provided detailed security analysis for these protocols. By substantially reducing the requirement on the round number for MPQC, our work paves the way towards practical multi-party quantum computation.

As direct applications, our work gives the first constant-round privacy-preserving schemes for
many natural distributed quantum computing tasks, including:
\begin{enumerate}
\item
\emph{Quantum fidelity evaluation:} Two parties wish to compute the fidelity between their two quantum states, but 
do not wish to reveal their states. 
\item
\emph{Quantum machine learning:} Multiple parties each has some labeled samples for quantum machine learning. However, they wish to learn the quantum model without revealing their private labeled samples which are costly to obtain.
\item
\emph{Quantum voting problem:} Each user $i$ generates a quantum state $q_i$ which is a superposition of $N$ candidates that he/she wishes to select where the amplitude of a candidate represents the user's inclination to this candidate. The users wish to jointly determine the candidate that has the largest fidelity with $\sum_i q_i$, without revealing their individual preferences.
Quantum voting is superior to classical voting, in the sense that the communication cost of quantum voting is reduced exponentially from $N$ to $\log N$
compared to its classical counterpart.
\item
\emph{Quantum pairing problem:} $2n$ graduate students are suddenly told that they need to pair themselves into $n$ two-person dormitory rooms. Each student $i$ has a quantum state $\rho_i$ that characterizes the various aspects of his personal habits, such as sleep time, tolerable noise level, etc.  Since a pair of students with similar habits may get along better, the students wish to maximize $\sum_{\{i,j\}\in P} F(\rho_i, \rho_j)$ where $P$ is the pairing. They also wish to hide the quantum state that describes their personal habits. 
\end{enumerate}

Our work opens a few interesting avenues for future research. First, our work requires the number of parties to be fixed. It remains to investigate whether a constant-round MPQC exists in the case that the number of parties increases with the problem size.
One possible route is to improve the protocol in \cite{dulek2020secure} so that its round number becomes independent of the number of parties, thereby removing the assumption in Corollary \ref{cor:multiparty} and making it an unconditional statement.
Secondly, the question of whether constant-round multi-party quantum computation is possible against malicious adversaries is still open. Intuitively, the answer to this question would be yes, as its classical counterpart has an affirmative answer. Thirdly, it would be fruitful to study concrete number of rounds for broadcast quantum channels and point-to-point quantum channels. In the classical case, it is known that two rounds suffice for multi-party computation with broadcast channels and three rounds suffice for point-to-point channels \cite{cohen2020broadcast}. Finally, an experimental demonstration for constant-round multi-party quantum computation would expedite its practical deployment, and is thus worth persuing. 

\vspace{0.5cm}
\subsect{Acknowledgements}
This work was supported by the internal Grant No. SLH00202007 from East China University of Science and Technology.

{\it Note added.}---After submission of the manuscript, we became aware of related works by Bartusek {\it et al.} \cite{bartusek2020round}  and Alon {\it et al.}  \cite{cryptoeprint:2020:1464}.

\vspace{1cm}
\noindent{{\Large\bf Materials and Methods}}\\

\subsect{Decomposable quantum random encoding}
\label{AppSec:QRE}
In this section, we review a DQRE scheme  \cite{brakerski2020quantum} for the convenience of the reader.
It is based on three ingredients: quantum computation via teleportation; twirling;
and group-randomizing DQRE.

First consider a circuit that only consists of Clifford gates.
We note the fact that the computation of the circuit can be performed through teleportation. 
To illustrate how this is possible, let us consider a circuit that consists of two 
gates $G_1$ and $G_2$, and the output of $G_1$ is the input of $G_2$. 
We can teleport the output of $G_1$ to the input of $G_2$ through teleportation using an EPR pair.
Recall that in teleportation, the result is masked by $X^aZ^b$ and hence needs to be corrected
according to the measurement outcomes. If no correction is made and $G_2$ is applied directly, the 
overall operation is then $G_2X^aZ^b$. If $G_2$ is a Clifford gate, we can defer the correction to the 
end as $G_2X^aZ^b = X^{a'}Z^{b'}G_2$. This circuit can be easily 
generalized to multiple gates, and the correction becomes 
\begin{eqnarray}
&&G_nX^{a_{n-1}}Z^{b_{n-1}}\cdots G_2X^{a_1}Z^{b_1}G_1  \nonumber \\
&=& X^{a_{fin}}Z^{b_{fin}}G_nG_{n-1}\cdots G_1.
\end{eqnarray}
  DQRE is then reduced to a classical RE of $a_{fin},b_{fin}$ from $a_1,b_1,\dots,a_{n-1},b_{n-1}$.

Next consider a general circuit, which additionally involves T gates. 
If $G_2$ is a T gate, then $G_2X^aZ^b = X^{a'}Z^{b'}P^{a'}G_2$. Hence, $P^{a'}$ needs to be applied
on the output of $G_2$ before another teleportation.  This can be done by a correction-and-teleport 
Clifford circuit $C$ that first makes the Clifford correction and then teleports the qubit. Two things remain
to be done. First, we need to find a DQRE for this correction-and-teleport circuit. Second, we need to 
turn the measurement step of the teleportation to a Clifford operation. The second issue can be dealt
with by performing $Z^s$ for a random $s$ as a substitute of measurement on the computational basis. This
is called twirling.

For the first issue, we perform the random encoding on $C(x)$ as $E=CR^{\dagger}$ and $z=R(x)$
for a random unitary $R$. It is easy to see that $E(z)=C(x)$. For efficient sampling, we can restrict $R$ to an element of the Clifford group. This encoding is called group-randomizing DQRE. Note that $z$ can be computed without
knowing the circuit $C$. In addition, $E$ can be computed through a classical procedure $g$. Say $C$ is described
by a classical function $f$ acting on some classical input $a$. The procedure $g$ takes $a$ and the classical description of $R$ as inputs. There are two steps of $g$. It first computes $C=f(a)$ and then computes $E=CR^{\dagger}$.
An illustration is shown in Fig.~\ref{fig:compiler}(B).

Finally, we need to post the labels that correspond to $a,b$ of the predecessor gates so that the correction 
can be performed. This can be done by including the labels of the teleportation correction $a,b$  of the current
gate in the output wire. The labels in the input wire will also be included in the output wire. This makes
the input and output length grow linearly with the circuit size. We note that this polynomial size input
does not mean that we need to apply a random Clifford on a polynomial number of qubits, since we only need to perform single-qubit randomizers (specifically from the conjugated PX group) in order 
to carry out the group-randomizing DQRE.

In summary, the essence for decoupling the gates and making the process constant depth is the use of EPR pairs. Each input qubit of a gate is half an EPR pair and each output qubit of a gate is also half an EPR pair. 

\vspace{0.5cm}
\subsect{Oblivious transfer against specious adversary}
Assume a quantum-secure public-key encryption scheme, Protocol~\ref{Fig:Procedure} realizes oblivious transfer against a specious adversary.

\begin{algorithm}
\caption{\textsc{Oblivious transfer}}
\begin{flushleft}
\textit{Input}: The two parties are called $\mathcal{S}$ and $\mathcal{R}$, where $\mathcal{S}$ has two strings $y_0, y_1\in \{0,1\}^n$, and $\mathcal{R}$ has a bit $b\in \{0,1\}$.  
\end{flushleft}
\begin{algorithmic}[1]
\STATE
$\mathcal{R}$ generates a pair of public and private keys $(pk, sk)$, and randomly samples a $pk'$ from the public-key space without knowing its corresponding private key. 
\vspace{0.2cm}
\STATE
If $b=0$, $\mathcal{R}$ sends $(pk, pk')$ to $\mathcal{S}$. Otherwise, $\mathcal{R}$ sends $(pk', pk)$ to $\mathcal{S}$.
\vspace{0.2cm}
\STATE
Upon receiving $(pk_0, pk_1)$, $\mathcal{S}$ returns $e_0 = \textrm{Enc}_{pk_0}(y_0),e_1 = \textrm{Enc}_{pk_1}(y_1)$ to $\mathcal{R}$.
\vspace{0.2cm}
\STATE On receiving $(e_0, e_1)$, $\mathcal{R}$ decrypts $e_b$ with $sk$.
\end{algorithmic}
\label{Fig:Procedure}
\end{algorithm}

This protocol is secure against a semi-honest adversary. Since only classical data is involved, this scheme is also secure against a specious adversary. 

\vspace{0.5cm}
\subsect{Classical multi-party computation scheme}
In this section, we review a classical constant-round secure multi-party computation protocol from \cite{beaver1990round,rogaway1991round} which is based on the GMW protocol \cite{micali1987play}. We will also mention the remark of \cite{tate2003garbled} which points out an error in \cite{beaver1990round,rogaway1991round}. We provide a proof for its constant round complexity, completing a missing piece in the literature.

We first begin with the GMW protocol which gives a secure multi-party computation protocol with a round number linear with the circuit depth. Any boolean circuit can be realized by AND, NOT, and XOR gates. Here, an AND gate has two inputs. A NOT gate has one input. A XOR gate has an unbounded number of inputs.  
The inputs belong to $n$ parties. For an input bit $b^\omega$ that belongs to party $i$, party $i$ generates $r_1^\omega,\cdots, r_{i-1}^\omega, r_{i+1}^\omega, \cdots, r_n^\omega$ and sends $r_j^\omega$ to party $j$. Party $i$ himself then holds $b^\omega \oplus r_1^\omega \oplus\cdots \oplus r_{i-1}^\omega \oplus r_{i+1}^\omega \oplus\cdots \oplus r_n^\omega$. Then each party shares a part of $b^\omega$ called $b^\omega_i$ and the XOR of these parts becomes $b^\omega$. We next show that each intermediate wire and each output wire have the same property. Then by XORing $n$ shares of each output wire, we obtain the circuit output.  

For a NOT gate that has input wire $\omega_1$ and output wire $\omega_2$, we take $b^{\omega_2}_1 = 1-b^{\omega_1}_1 $ and $b^{\omega_2}_j = b^{\omega_1}_j$ for $2\le j \le n$. For a XOR gate with input wires $\omega_1,\cdots \omega_K$ and  an output wire $\omega_0$, we take $b^{\omega_0}_i =    b^{\omega_1}_i  \oplus\cdots \oplus b^{\omega_K}_i$ for $1\le i \le n$. For an AND gate, let the input wires be $\omega_1$ and $\omega_2$, and let the output wire be $\omega_3$.   Let us first examine the case of two parties, called $P_1$ and $P_2$.  According to the functionality of the AND gate, we have 
\begin{equation}
b^{\omega_3}  = ( b^{\omega_1}_1  \oplus  b^{\omega_1}_2  )  \wedge  ( b^{\omega_2}_1  \oplus  b^{\omega_2}_2  ). 
\end{equation}
From the view of $P_1$, he does not know $b^{\omega_1}_2$ and $b^{\omega_2}_2$, hence he views the expression as a function $S(b^{\omega_1}_2,b^{\omega_2}_2)= ( b^{\omega_1}_1  \oplus  b^{\omega_1}_2  )  \wedge  ( b^{\omega_2}_1  \oplus  b^{\omega_2}_2  )$. He then chooses a random bit $r$, and runs a 1-out-of-4 OT with $P_2$ with four values $(r\oplus S(0,0), r\oplus S(0,1), r\oplus S(1,0), r\oplus S(1,0))$. By the property of $OT$, $P_2$ gets $ r\oplus S(b^{\omega_1}_2,b^{\omega_2}_2) = r\oplus b^{\omega_3}$ and takes this to be $b^{\omega_3}_2$. $P_1$ takes $r$ to be $b^{\omega_3}_1$. For $n$ parties, we note that
\begin{eqnarray}
b^{\omega_3}  & =& ( b^{\omega_1}_1  \oplus \cdots \oplus b^{\omega_1}_n  )  \wedge  ( b^{\omega_2}_1  \oplus \cdots \oplus b^{\omega_2}_n  )  \nonumber \\
 & =& ( \oplus_{i=1}^n b^{\omega_1}_i  \wedge  b^{\omega_2}_i) \oplus (  \oplus_{i\not=j} b^{\omega_1}_i  \wedge  b^{\omega_2}_j)  
\end{eqnarray}
which is a 2-depth circuit with two-party AND gates and XOR gates. This finishes the GMW protocol.

Now we turn to the scheme of \cite{beaver1990round,rogaway1991round}. It consists of two parts. The first part is generating \emph{gate labels} and \emph{input signals}, which are accessible to all parties. The second part is evaluating gate labels and input signals  by each party. The second part involves no communication between the parties while the first part involves constant rounds of communication.

Let us first define the setting. Let $\Sigma=\{0,1\}$. Each party $i$ has an $\ell$ bit input $x_i \in \Sigma^{\ell}$ and a $2kW+W-l$ bit random string $r_i \in \Sigma^{2kW+W-l}$. In the first part, the parties jointly compute four gate labels $A_{00}^g, A_{10}^g, A_{01}^g, A_{11}^g$ for each gate $g$ and an input signal $\sigma^\omega$ for each input bit $\omega$. 

Each $r_i$ can be expressed as $s^1_{0i}s^1_{1i}\cdots s^W_{0i}s^W_{1i}\lambda_i^1\cdots\lambda_i^{W-l}$ where $W$ is the number of wires, $s$'s are of length $k$ and $\lambda$'s are of length 1. 
The mask on the semantic value is
\begin{equation}
 \lambda^{\omega} =  \lambda^{\omega}_1 \oplus \cdots \oplus \lambda^{\omega}_n
\end{equation}
for non-output wires (i.e., $\omega \le W-l$) and $ \lambda^{\omega} =0$ for output wires (i.e., $W-l < \omega \le W$). 
Hence, the input wires from the parties (i.e., $\omega \le n\ell$ ) satisfy
\begin{equation}
\sigma^\omega  = s^\omega_{b^\omega \oplus \lambda^\omega},
\end{equation}
where $b^\omega$ is the semantic value of the wire $\omega$ and $s^\omega_b$ is given by $s_{b1}^\omega \cdots s_{bn}^\omega b$.

The gate labels for a gate $g$ are computed as 
\begin{eqnarray}
A_{ab}^g &=& G_b(s_{a1}^\alpha) \oplus \cdots \oplus G_b(s_{an}^\alpha) \oplus G_a(s_{b1}^\beta) \oplus \cdots \nonumber \\ &&\oplus G_a(s_{bn}^\beta)
 \oplus s^{\gamma}_{[(\lambda^\alpha\oplus a) \otimes( \lambda^\beta \oplus b)]\oplus \lambda^\gamma},
\end{eqnarray}
where $\otimes$ is the function computed by $g$. Here, $a,b\in \{0,1\}$, $G_a(\cdot)$ and $G_b(\cdot)$ are pseudo-random generators from $\Sigma^k$ to $\Sigma^{nk+1}$. 

Each party can compute the output of their joint computation on its own given the gate labels and the input signals.
Starting from the input signals, one can compute intermediate signals and output signals as follows. For each gate $g$ with input wires $\alpha,\beta$ and an output wire $\gamma$, 
one computes
\begin{equation}
\sigma^\gamma = G_b(\sigma_1^\alpha) \oplus \cdots \oplus G_b(\sigma_n^\alpha) \oplus G_a(\sigma_1^\beta) \oplus \cdots \oplus G_a(\sigma_n^\beta) \oplus A_{ab}^g
\end{equation}
Here $\sigma^\omega_i$ is a bit string defined by the $(i-1)k+1$- to $ik$- bits of $\sigma^\omega$. The bits $a$ and $b$ are the last bits of $\sigma^\alpha$ and $\sigma^\beta$, respectively.

Finally, the least significant bits of the output wires $\sigma^{W-l+1},\cdots, \sigma^W$ are outputted.

This finishes the description of the scheme in  \cite{beaver1990round,rogaway1991round}. In the original description of \cite{beaver1990round,rogaway1991round}, the authors make a false claim that a wire can be used for multiple times as inputs to multiple gates. It is shown in \cite{tate2003garbled} that this will raise a security loophole for the scheme. Hence, we abandon such a false claim in our description. Moreover, for the quantum setting, due to no-cloning theorem, multiple uses of a quantum wire is impossible.

We are now ready to show that this process requires only constant rounds of communication. To the best of our knowledge, the proof that we show below is new in the literature. It suffices to show that the process can be expressed as a constant-depth circuit with the gate set compatible with the GMW protocol. The first quantity that requires secure joint computation of $n$ parties is $\lambda^\omega$. From its expression, it can be realized by a single XOR gate and hence requires a circuit depth 1. 

The second quantity that requires joint computation is $\sigma^\omega$. It can be decomposed as two parts, namely $b=b^\omega \oplus \lambda^\omega$ and $\sigma^\omega= s^\omega_b$. The first part can be realized by a single XOR gate. For the second part, we notice that each bit of $\sigma^\omega$ can be computed separately. In more details, let $(\sigma^\omega)_j$ denote the $j$-th bit of $\sigma^\omega$ and let $(s^\omega_b)_j$ denote the $j$-th bit of $s^\omega_b$. Then $(\sigma^\omega)_j$ is only determined by $(s^\omega_0)_j$, $(s^\omega_1)_j$ and $b$. More precisely, for the last bit $j$, we have $(\sigma^\omega)_j=b$ and for all other $j$'s, we have 
\begin{equation}
(\sigma^\omega)_j=(1-b) \wedge (s^\omega_0)_j \oplus b \wedge (s^\omega_1)_j,
\end{equation}
which can be realized by a depth-2 circuit. Taking into account that $b= b^\omega \oplus \lambda^\omega$ requires a depth-1 circuit and the fact that  $\lambda^\omega$ itself requires a depth-1 circuit, the second quantity requires a circuit of depth at most 4. Further optimizing the circuit depth is certainly possible, but we will not pursue it here.  

The third quantity that requires joint computation is $A_{ab}^g$. Note that $s_{ai}^\alpha$ and $s_{bi}^\beta$ are held by party $i$ initially, hence he can locally compute $G_b(s_{ai}^\alpha)$ and $G_a(s_{bi}^\beta)$. Hence, the joint computation only involves the XOR operation and the computation of $s^\gamma_c$. Here $c=[(\lambda^\alpha\oplus a) \otimes( \lambda^\beta \oplus b)]\oplus \lambda^\gamma$. By its expression, $c$ can be computed by a depth-3 circuit. $s^\gamma_c$ can be computed by a depth-2 circuit following the analysis of the second quantity. The final XORing of some $G(\cdots)$'s and $s^\gamma_c$ requires a single XOR gate. In summary, a depth-six circuit suffices for the joint computation of $A_{ab}^g$.

In summary, the gate labels and input signals require at most a depth-6 circuit to compute. According to the GMW protocol, a constant-depth circuit requires only a constant number of communication. This finishes the proof of constant round complexity.

\bibliographystyle{apsrev4-1}

\bibliography{BibliMPQC}

\begin{thebibliography}{47}%
\makeatletter
\providecommand \@ifxundefined [1]{%
 \@ifx{#1\undefined}
}%
\providecommand \@ifnum [1]{%
 \ifnum #1\expandafter \@firstoftwo
 \else \expandafter \@secondoftwo
 \fi
}%
\providecommand \@ifx [1]{%
 \ifx #1\expandafter \@firstoftwo
 \else \expandafter \@secondoftwo
 \fi
}%
\providecommand \natexlab [1]{#1}%
\providecommand \enquote  [1]{``#1''}%
\providecommand \bibnamefont  [1]{#1}%
\providecommand \bibfnamefont [1]{#1}%
\providecommand \citenamefont [1]{#1}%
\providecommand \href@noop [0]{\@secondoftwo}%
\providecommand \href [0]{\begingroup \@sanitize@url \@href}%
\providecommand \@href[1]{\@@startlink{#1}\@@href}%
\providecommand \@@href[1]{\endgroup#1\@@endlink}%
\providecommand \@sanitize@url [0]{\catcode `\\12\catcode `\$12\catcode
  `\&12\catcode `\#12\catcode `\^12\catcode `\_12\catcode `\%12\relax}%
\providecommand \@@startlink[1]{}%
\providecommand \@@endlink[0]{}%
\providecommand \url  [0]{\begingroup\@sanitize@url \@url }%
\providecommand \@url [1]{\endgroup\@href {#1}{\urlprefix }}%
\providecommand \urlprefix  [0]{URL }%
\providecommand \Eprint [0]{\href }%
\providecommand \doibase [0]{http://dx.doi.org/}%
\providecommand \selectlanguage [0]{\@gobble}%
\providecommand \bibinfo  [0]{\@secondoftwo}%
\providecommand \bibfield  [0]{\@secondoftwo}%
\providecommand \translation [1]{[#1]}%
\providecommand \BibitemOpen [0]{}%
\providecommand \bibitemStop [0]{}%
\providecommand \bibitemNoStop [0]{.\EOS\space}%
\providecommand \EOS [0]{\spacefactor3000\relax}%
\providecommand \BibitemShut  [1]{\csname bibitem#1\endcsname}%
\let\auto@bib@innerbib\@empty
\bibitem [{\citenamefont {Bennett}\ and\ \citenamefont
  {Brassard}(1984)}]{Bennett:BB84:1984}%
  \BibitemOpen
  \bibfield  {author} {\bibinfo {author} {\bibfnamefont {C.~H.}\ \bibnamefont
  {Bennett}}\ and\ \bibinfo {author} {\bibfnamefont {G.}~\bibnamefont
  {Brassard}},\ }in\ \href@noop {} {\emph {\bibinfo {booktitle} {Proceedings of
  the IEEE International Conference on Computers, Systems and Signal
  Processing}}}\ (\bibinfo  {publisher} {IEEE},\ \bibinfo {address} {New
  York},\ \bibinfo {year} {1984})\ pp.\ \bibinfo {pages} {175--179}\BibitemShut
  {NoStop}%
\bibitem [{\citenamefont {Lo}\ and\ \citenamefont {Chau}(1999)}]{LoChauQKD_99}%
  \BibitemOpen
  \bibfield  {author} {\bibinfo {author} {\bibfnamefont {H.-K.}\ \bibnamefont
  {Lo}}\ and\ \bibinfo {author} {\bibfnamefont {H.~F.}\ \bibnamefont {Chau}},\
  }\href@noop {} {\bibfield  {journal} {\bibinfo  {journal} {Science}\ }\textbf
  {\bibinfo {volume} {283}},\ \bibinfo {pages} {2050} (\bibinfo {year}
  {1999})}\BibitemShut {NoStop}%
\bibitem [{\citenamefont {Shor}\ and\ \citenamefont
  {Preskill}(2000)}]{ShorPreskill_00}%
  \BibitemOpen
  \bibfield  {author} {\bibinfo {author} {\bibfnamefont {P.~W.}\ \bibnamefont
  {Shor}}\ and\ \bibinfo {author} {\bibfnamefont {J.}~\bibnamefont
  {Preskill}},\ }\href@noop {} {\bibfield  {journal} {\bibinfo  {journal}
  {Phys.~Rev.~Lett.~}\ }\textbf {\bibinfo {volume} {85}},\ \bibinfo {pages}
  {441} (\bibinfo {year} {2000})}\BibitemShut {NoStop}%
\bibitem [{\citenamefont {Lo}\ \emph {et~al.}(2012)\citenamefont {Lo},
  \citenamefont {Curty},\ and\ \citenamefont {Qi}}]{Lo:MDIQKD:2012}%
  \BibitemOpen
  \bibfield  {author} {\bibinfo {author} {\bibfnamefont {H.-K.}\ \bibnamefont
  {Lo}}, \bibinfo {author} {\bibfnamefont {M.}~\bibnamefont {Curty}}, \ and\
  \bibinfo {author} {\bibfnamefont {B.}~\bibnamefont {Qi}},\ }\href@noop {}
  {\bibfield  {journal} {\bibinfo  {journal} {Phys. Rev. Lett.}\ }\textbf
  {\bibinfo {volume} {108}},\ \bibinfo {pages} {130503} (\bibinfo {year}
  {2012})}\BibitemShut {NoStop}%
\bibitem [{\citenamefont {Arnon-Friedman}\ \emph {et~al.}(2018)\citenamefont
  {Arnon-Friedman}, \citenamefont {Dupuis}, \citenamefont {Fawzi},
  \citenamefont {Renner},\ and\ \citenamefont {Vidick}}]{arnon2018practical}%
  \BibitemOpen
  \bibfield  {author} {\bibinfo {author} {\bibfnamefont {R.}~\bibnamefont
  {Arnon-Friedman}}, \bibinfo {author} {\bibfnamefont {F.}~\bibnamefont
  {Dupuis}}, \bibinfo {author} {\bibfnamefont {O.}~\bibnamefont {Fawzi}},
  \bibinfo {author} {\bibfnamefont {R.}~\bibnamefont {Renner}}, \ and\ \bibinfo
  {author} {\bibfnamefont {T.}~\bibnamefont {Vidick}},\ }\href@noop {}
  {\bibfield  {journal} {\bibinfo  {journal} {Nature communications}\ }\textbf
  {\bibinfo {volume} {9}},\ \bibinfo {pages} {1} (\bibinfo {year}
  {2018})}\BibitemShut {NoStop}%
\bibitem [{\citenamefont {Liao}\ \emph {et~al.}(2017)\citenamefont {Liao},
  \citenamefont {Cai}, \citenamefont {Liu}, \citenamefont {Zhang},
  \citenamefont {Li}, \citenamefont {Ren}, \citenamefont {Yin}, \citenamefont
  {Shen}, \citenamefont {Cao}, \citenamefont {Li} \emph
  {et~al.}}]{liao2017satellite}%
  \BibitemOpen
  \bibfield  {author} {\bibinfo {author} {\bibfnamefont {S.-K.}\ \bibnamefont
  {Liao}}, \bibinfo {author} {\bibfnamefont {W.-Q.}\ \bibnamefont {Cai}},
  \bibinfo {author} {\bibfnamefont {W.-Y.}\ \bibnamefont {Liu}}, \bibinfo
  {author} {\bibfnamefont {L.}~\bibnamefont {Zhang}}, \bibinfo {author}
  {\bibfnamefont {Y.}~\bibnamefont {Li}}, \bibinfo {author} {\bibfnamefont
  {J.-G.}\ \bibnamefont {Ren}}, \bibinfo {author} {\bibfnamefont
  {J.}~\bibnamefont {Yin}}, \bibinfo {author} {\bibfnamefont {Q.}~\bibnamefont
  {Shen}}, \bibinfo {author} {\bibfnamefont {Y.}~\bibnamefont {Cao}}, \bibinfo
  {author} {\bibfnamefont {Z.-P.}\ \bibnamefont {Li}},  \emph {et~al.},\
  }\href@noop {} {\bibfield  {journal} {\bibinfo  {journal} {Nature}\ }\textbf
  {\bibinfo {volume} {549}},\ \bibinfo {pages} {43} (\bibinfo {year}
  {2017})}\BibitemShut {NoStop}%
\bibitem [{\citenamefont {Fu}\ \emph {et~al.}(2015)\citenamefont {Fu},
  \citenamefont {Yin}, \citenamefont {Chen},\ and\ \citenamefont
  {Chen}}]{fu2015long}%
  \BibitemOpen
  \bibfield  {author} {\bibinfo {author} {\bibfnamefont {Y.}~\bibnamefont
  {Fu}}, \bibinfo {author} {\bibfnamefont {H.-L.}\ \bibnamefont {Yin}},
  \bibinfo {author} {\bibfnamefont {T.-Y.}\ \bibnamefont {Chen}}, \ and\
  \bibinfo {author} {\bibfnamefont {Z.-B.}\ \bibnamefont {Chen}},\ }\href@noop
  {} {\bibfield  {journal} {\bibinfo  {journal} {Physical review letters}\
  }\textbf {\bibinfo {volume} {114}},\ \bibinfo {pages} {090501} (\bibinfo
  {year} {2015})}\BibitemShut {NoStop}%
\bibitem [{\citenamefont {He}\ and\ \citenamefont
  {Reid}(2013)}]{he2013genuine}%
  \BibitemOpen
  \bibfield  {author} {\bibinfo {author} {\bibfnamefont {Q.}~\bibnamefont
  {He}}\ and\ \bibinfo {author} {\bibfnamefont {M.}~\bibnamefont {Reid}},\
  }\href@noop {} {\bibfield  {journal} {\bibinfo  {journal} {Physical Review
  Letters}\ }\textbf {\bibinfo {volume} {111}},\ \bibinfo {pages} {250403}
  (\bibinfo {year} {2013})}\BibitemShut {NoStop}%
\bibitem [{\citenamefont {Yonezawa}\ \emph {et~al.}(2004)\citenamefont
  {Yonezawa}, \citenamefont {Aoki},\ and\ \citenamefont
  {Furusawa}}]{yonezawa2004demonstration}%
  \BibitemOpen
  \bibfield  {author} {\bibinfo {author} {\bibfnamefont {H.}~\bibnamefont
  {Yonezawa}}, \bibinfo {author} {\bibfnamefont {T.}~\bibnamefont {Aoki}}, \
  and\ \bibinfo {author} {\bibfnamefont {A.}~\bibnamefont {Furusawa}},\
  }\href@noop {} {\bibfield  {journal} {\bibinfo  {journal} {Nature}\ }\textbf
  {\bibinfo {volume} {431}},\ \bibinfo {pages} {430} (\bibinfo {year}
  {2004})}\BibitemShut {NoStop}%
\bibitem [{\citenamefont {Jing}\ \emph {et~al.}(2003)\citenamefont {Jing},
  \citenamefont {Zhang}, \citenamefont {Yan}, \citenamefont {Zhao},
  \citenamefont {Xie},\ and\ \citenamefont {Peng}}]{jing2003experimental}%
  \BibitemOpen
  \bibfield  {author} {\bibinfo {author} {\bibfnamefont {J.}~\bibnamefont
  {Jing}}, \bibinfo {author} {\bibfnamefont {J.}~\bibnamefont {Zhang}},
  \bibinfo {author} {\bibfnamefont {Y.}~\bibnamefont {Yan}}, \bibinfo {author}
  {\bibfnamefont {F.}~\bibnamefont {Zhao}}, \bibinfo {author} {\bibfnamefont
  {C.}~\bibnamefont {Xie}}, \ and\ \bibinfo {author} {\bibfnamefont
  {K.}~\bibnamefont {Peng}},\ }\href@noop {} {\bibfield  {journal} {\bibinfo
  {journal} {Physical review letters}\ }\textbf {\bibinfo {volume} {90}},\
  \bibinfo {pages} {167903} (\bibinfo {year} {2003})}\BibitemShut {NoStop}%
\bibitem [{\citenamefont {Gisin}\ \emph {et~al.}(2020)\citenamefont {Gisin},
  \citenamefont {Bancal}, \citenamefont {Cai}, \citenamefont {Remy},
  \citenamefont {Tavakoli}, \citenamefont {Cruzeiro}, \citenamefont {Popescu},\
  and\ \citenamefont {Brunner}}]{gisin2020constraints}%
  \BibitemOpen
  \bibfield  {author} {\bibinfo {author} {\bibfnamefont {N.}~\bibnamefont
  {Gisin}}, \bibinfo {author} {\bibfnamefont {J.-D.}\ \bibnamefont {Bancal}},
  \bibinfo {author} {\bibfnamefont {Y.}~\bibnamefont {Cai}}, \bibinfo {author}
  {\bibfnamefont {P.}~\bibnamefont {Remy}}, \bibinfo {author} {\bibfnamefont
  {A.}~\bibnamefont {Tavakoli}}, \bibinfo {author} {\bibfnamefont {E.~Z.}\
  \bibnamefont {Cruzeiro}}, \bibinfo {author} {\bibfnamefont {S.}~\bibnamefont
  {Popescu}}, \ and\ \bibinfo {author} {\bibfnamefont {N.}~\bibnamefont
  {Brunner}},\ }\href@noop {} {\bibfield  {journal} {\bibinfo  {journal}
  {Nature communications}\ }\textbf {\bibinfo {volume} {11}},\ \bibinfo {pages}
  {1} (\bibinfo {year} {2020})}\BibitemShut {NoStop}%
\bibitem [{\citenamefont {Spiller}\ \emph {et~al.}(2006)\citenamefont
  {Spiller}, \citenamefont {Nemoto}, \citenamefont {Braunstein}, \citenamefont
  {Munro}, \citenamefont {van Loock},\ and\ \citenamefont
  {Milburn}}]{spiller2006quantum}%
  \BibitemOpen
  \bibfield  {author} {\bibinfo {author} {\bibfnamefont {T.~P.}\ \bibnamefont
  {Spiller}}, \bibinfo {author} {\bibfnamefont {K.}~\bibnamefont {Nemoto}},
  \bibinfo {author} {\bibfnamefont {S.~L.}\ \bibnamefont {Braunstein}},
  \bibinfo {author} {\bibfnamefont {W.~J.}\ \bibnamefont {Munro}}, \bibinfo
  {author} {\bibfnamefont {P.}~\bibnamefont {van Loock}}, \ and\ \bibinfo
  {author} {\bibfnamefont {G.~J.}\ \bibnamefont {Milburn}},\ }\href@noop {}
  {\bibfield  {journal} {\bibinfo  {journal} {New Journal of Physics}\ }\textbf
  {\bibinfo {volume} {8}},\ \bibinfo {pages} {30} (\bibinfo {year}
  {2006})}\BibitemShut {NoStop}%
\bibitem [{\citenamefont {Giovannetti}\ \emph {et~al.}(2004)\citenamefont
  {Giovannetti}, \citenamefont {Lloyd},\ and\ \citenamefont
  {Maccone}}]{giovannetti2004quantum}%
  \BibitemOpen
  \bibfield  {author} {\bibinfo {author} {\bibfnamefont {V.}~\bibnamefont
  {Giovannetti}}, \bibinfo {author} {\bibfnamefont {S.}~\bibnamefont {Lloyd}},
  \ and\ \bibinfo {author} {\bibfnamefont {L.}~\bibnamefont {Maccone}},\
  }\href@noop {} {\bibfield  {journal} {\bibinfo  {journal} {Science}\ }\textbf
  {\bibinfo {volume} {306}},\ \bibinfo {pages} {1330} (\bibinfo {year}
  {2004})}\BibitemShut {NoStop}%
\bibitem [{\citenamefont {Giovannetti}\ \emph {et~al.}(2001)\citenamefont
  {Giovannetti}, \citenamefont {Lloyd},\ and\ \citenamefont
  {Maccone}}]{giovannetti2001quantum}%
  \BibitemOpen
  \bibfield  {author} {\bibinfo {author} {\bibfnamefont {V.}~\bibnamefont
  {Giovannetti}}, \bibinfo {author} {\bibfnamefont {S.}~\bibnamefont {Lloyd}},
  \ and\ \bibinfo {author} {\bibfnamefont {L.}~\bibnamefont {Maccone}},\
  }\href@noop {} {\bibfield  {journal} {\bibinfo  {journal} {Nature}\ }\textbf
  {\bibinfo {volume} {412}},\ \bibinfo {pages} {417} (\bibinfo {year}
  {2001})}\BibitemShut {NoStop}%
\bibitem [{\citenamefont {Cirac}\ \emph {et~al.}(1999)\citenamefont {Cirac},
  \citenamefont {Ekert}, \citenamefont {Huelga},\ and\ \citenamefont
  {Macchiavello}}]{cirac1999distributed}%
  \BibitemOpen
  \bibfield  {author} {\bibinfo {author} {\bibfnamefont {J.}~\bibnamefont
  {Cirac}}, \bibinfo {author} {\bibfnamefont {A.}~\bibnamefont {Ekert}},
  \bibinfo {author} {\bibfnamefont {S.}~\bibnamefont {Huelga}}, \ and\ \bibinfo
  {author} {\bibfnamefont {C.}~\bibnamefont {Macchiavello}},\ }\href@noop {}
  {\bibfield  {journal} {\bibinfo  {journal} {Physical Review A}\ }\textbf
  {\bibinfo {volume} {59}},\ \bibinfo {pages} {4249} (\bibinfo {year}
  {1999})}\BibitemShut {NoStop}%
\bibitem [{\citenamefont {Gentry}\ and\ \citenamefont
  {Boneh}(2009)}]{gentry2009fully}%
  \BibitemOpen
  \bibfield  {author} {\bibinfo {author} {\bibfnamefont {C.}~\bibnamefont
  {Gentry}}\ and\ \bibinfo {author} {\bibfnamefont {D.}~\bibnamefont {Boneh}},\
  }\href@noop {} {\emph {\bibinfo {title} {A fully homomorphic encryption
  scheme}}},\ Vol.~\bibinfo {volume} {20}\ (\bibinfo  {publisher} {Stanford
  university Stanford},\ \bibinfo {year} {2009})\BibitemShut {NoStop}%
\bibitem [{\citenamefont {Blum}(1983)}]{blum1983coin}%
  \BibitemOpen
  \bibfield  {author} {\bibinfo {author} {\bibfnamefont {M.}~\bibnamefont
  {Blum}},\ }\href@noop {} {\bibfield  {journal} {\bibinfo  {journal} {ACM
  SIGACT News}\ }\textbf {\bibinfo {volume} {15}},\ \bibinfo {pages} {23}
  (\bibinfo {year} {1983})}\BibitemShut {NoStop}%
\bibitem [{\citenamefont {Rabin}(2005)}]{rabin2005exchange}%
  \BibitemOpen
  \bibfield  {author} {\bibinfo {author} {\bibfnamefont {M.~O.}\ \bibnamefont
  {Rabin}},\ }\href@noop {} {\bibfield  {journal} {\bibinfo  {journal} {IACR
  Cryptol. ePrint Arch.}\ }\textbf {\bibinfo {volume} {2005}} (\bibinfo {year}
  {2005})}\BibitemShut {NoStop}%
\bibitem [{\citenamefont {Naor}(1991)}]{naor1991bit}%
  \BibitemOpen
  \bibfield  {author} {\bibinfo {author} {\bibfnamefont {M.}~\bibnamefont
  {Naor}},\ }\href@noop {} {\bibfield  {journal} {\bibinfo  {journal} {Journal
  of cryptology}\ }\textbf {\bibinfo {volume} {4}},\ \bibinfo {pages} {151}
  (\bibinfo {year} {1991})}\BibitemShut {NoStop}%
\bibitem [{\citenamefont {Peacock}\ \emph {et~al.}(2004)\citenamefont
  {Peacock}, \citenamefont {Ke},\ and\ \citenamefont
  {Wilkerson}}]{peacock2004typing}%
  \BibitemOpen
  \bibfield  {author} {\bibinfo {author} {\bibfnamefont {A.}~\bibnamefont
  {Peacock}}, \bibinfo {author} {\bibfnamefont {X.}~\bibnamefont {Ke}}, \ and\
  \bibinfo {author} {\bibfnamefont {M.}~\bibnamefont {Wilkerson}},\ }\href@noop
  {} {\bibfield  {journal} {\bibinfo  {journal} {IEEE Security \& Privacy}\
  }\textbf {\bibinfo {volume} {2}},\ \bibinfo {pages} {40} (\bibinfo {year}
  {2004})}\BibitemShut {NoStop}%
\bibitem [{\citenamefont {Bellare}\ \emph {et~al.}(2000)\citenamefont
  {Bellare}, \citenamefont {Pointcheval},\ and\ \citenamefont
  {Rogaway}}]{bellare2000authenticated}%
  \BibitemOpen
  \bibfield  {author} {\bibinfo {author} {\bibfnamefont {M.}~\bibnamefont
  {Bellare}}, \bibinfo {author} {\bibfnamefont {D.}~\bibnamefont
  {Pointcheval}}, \ and\ \bibinfo {author} {\bibfnamefont {P.}~\bibnamefont
  {Rogaway}},\ }in\ \href@noop {} {\emph {\bibinfo {booktitle} {International
  conference on the theory and applications of cryptographic techniques}}}\
  (\bibinfo {organization} {Springer},\ \bibinfo {year} {2000})\ pp.\ \bibinfo
  {pages} {139--155}\BibitemShut {NoStop}%
\bibitem [{\citenamefont {Feige}\ \emph {et~al.}(1988)\citenamefont {Feige},
  \citenamefont {Fiat},\ and\ \citenamefont {Shamir}}]{feige1988zero}%
  \BibitemOpen
  \bibfield  {author} {\bibinfo {author} {\bibfnamefont {U.}~\bibnamefont
  {Feige}}, \bibinfo {author} {\bibfnamefont {A.}~\bibnamefont {Fiat}}, \ and\
  \bibinfo {author} {\bibfnamefont {A.}~\bibnamefont {Shamir}},\ }\href@noop {}
  {\bibfield  {journal} {\bibinfo  {journal} {Journal of cryptology}\ }\textbf
  {\bibinfo {volume} {1}},\ \bibinfo {pages} {77} (\bibinfo {year}
  {1988})}\BibitemShut {NoStop}%
\bibitem [{\citenamefont {Chellappa}\ and\ \citenamefont
  {Pavlou}(2002)}]{chellappa2002perceived}%
  \BibitemOpen
  \bibfield  {author} {\bibinfo {author} {\bibfnamefont {R.~K.}\ \bibnamefont
  {Chellappa}}\ and\ \bibinfo {author} {\bibfnamefont {P.~A.}\ \bibnamefont
  {Pavlou}},\ }\href@noop {} {\bibfield  {journal} {\bibinfo  {journal}
  {Logistics Information Management}\ } (\bibinfo {year} {2002})}\BibitemShut
  {NoStop}%
\bibitem [{\citenamefont {Riley}\ and\ \citenamefont
  {Samuelson}(1981)}]{riley1981optimal}%
  \BibitemOpen
  \bibfield  {author} {\bibinfo {author} {\bibfnamefont {J.~G.}\ \bibnamefont
  {Riley}}\ and\ \bibinfo {author} {\bibfnamefont {W.~F.}\ \bibnamefont
  {Samuelson}},\ }\href@noop {} {\bibfield  {journal} {\bibinfo  {journal} {The
  American Economic Review}\ }\textbf {\bibinfo {volume} {71}},\ \bibinfo
  {pages} {381} (\bibinfo {year} {1981})}\BibitemShut {NoStop}%
\bibitem [{\citenamefont {Garay}\ \emph {et~al.}(1999)\citenamefont {Garay},
  \citenamefont {Jakobsson},\ and\ \citenamefont {MacKenzie}}]{garay1999abuse}%
  \BibitemOpen
  \bibfield  {author} {\bibinfo {author} {\bibfnamefont {J.~A.}\ \bibnamefont
  {Garay}}, \bibinfo {author} {\bibfnamefont {M.}~\bibnamefont {Jakobsson}}, \
  and\ \bibinfo {author} {\bibfnamefont {P.}~\bibnamefont {MacKenzie}},\ }in\
  \href@noop {} {\emph {\bibinfo {booktitle} {Annual International Cryptology
  Conference}}}\ (\bibinfo {organization} {Springer},\ \bibinfo {year} {1999})\
  pp.\ \bibinfo {pages} {449--466}\BibitemShut {NoStop}%
\bibitem [{\citenamefont {Yao}(1986)}]{yao1986generate}%
  \BibitemOpen
  \bibfield  {author} {\bibinfo {author} {\bibfnamefont {A.~C.-C.}\
  \bibnamefont {Yao}},\ }in\ \href@noop {} {\emph {\bibinfo {booktitle} {27th
  Annual Symposium on Foundations of Computer Science (sfcs 1986)}}}\ (\bibinfo
  {organization} {IEEE},\ \bibinfo {year} {1986})\ pp.\ \bibinfo {pages}
  {162--167}\BibitemShut {NoStop}%
\bibitem [{\citenamefont {Goldreich}\ \emph {et~al.}(1987)\citenamefont
  {Goldreich}, \citenamefont {Micali},\ and\ \citenamefont
  {Wigderson}}]{micali1987play}%
  \BibitemOpen
  \bibfield  {author} {\bibinfo {author} {\bibfnamefont {O.}~\bibnamefont
  {Goldreich}}, \bibinfo {author} {\bibfnamefont {S.}~\bibnamefont {Micali}}, \
  and\ \bibinfo {author} {\bibfnamefont {A.}~\bibnamefont {Wigderson}},\ }in\
  \href@noop {} {\emph {\bibinfo {booktitle} {Proceedings of the Nineteenth ACM
  Symp. on Theory of Computing, STOC}}}\ (\bibinfo {year} {1987})\ pp.\
  \bibinfo {pages} {218--229}\BibitemShut {NoStop}%
\bibitem [{\citenamefont {Beaver}\ \emph {et~al.}(1990)\citenamefont {Beaver},
  \citenamefont {Micali},\ and\ \citenamefont {Rogaway}}]{beaver1990round}%
  \BibitemOpen
  \bibfield  {author} {\bibinfo {author} {\bibfnamefont {D.}~\bibnamefont
  {Beaver}}, \bibinfo {author} {\bibfnamefont {S.}~\bibnamefont {Micali}}, \
  and\ \bibinfo {author} {\bibfnamefont {P.}~\bibnamefont {Rogaway}},\ }in\
  \href@noop {} {\emph {\bibinfo {booktitle} {Proceedings of the twenty-second
  annual ACM symposium on Theory of computing}}}\ (\bibinfo {year} {1990})\
  pp.\ \bibinfo {pages} {503--513}\BibitemShut {NoStop}%
\bibitem [{\citenamefont {Rogaway}(1991)}]{rogaway1991round}%
  \BibitemOpen
  \bibfield  {author} {\bibinfo {author} {\bibfnamefont {P.}~\bibnamefont
  {Rogaway}},\ }\emph {\bibinfo {title} {The round complexity of secure
  protocols}},\ \href@noop {} {Ph.D. thesis},\ \bibinfo  {school}
  {Massachusetts Institute of Technology} (\bibinfo {year} {1991})\BibitemShut
  {NoStop}%
\bibitem [{\citenamefont {Garg}\ and\ \citenamefont
  {Srinivasan}(2018)}]{garg2018two}%
  \BibitemOpen
  \bibfield  {author} {\bibinfo {author} {\bibfnamefont {S.}~\bibnamefont
  {Garg}}\ and\ \bibinfo {author} {\bibfnamefont {A.}~\bibnamefont
  {Srinivasan}},\ }in\ \href@noop {} {\emph {\bibinfo {booktitle} {Annual
  International Conference on the Theory and Applications of Cryptographic
  Techniques}}}\ (\bibinfo {organization} {Springer},\ \bibinfo {year} {2018})\
  pp.\ \bibinfo {pages} {468--499}\BibitemShut {NoStop}%
\bibitem [{\citenamefont {Benhamouda}\ and\ \citenamefont
  {Lin}(2018)}]{benhamouda2018k}%
  \BibitemOpen
  \bibfield  {author} {\bibinfo {author} {\bibfnamefont {F.}~\bibnamefont
  {Benhamouda}}\ and\ \bibinfo {author} {\bibfnamefont {H.}~\bibnamefont
  {Lin}},\ }in\ \href@noop {} {\emph {\bibinfo {booktitle} {Annual
  International Conference on the Theory and Applications of Cryptographic
  Techniques}}}\ (\bibinfo {organization} {Springer},\ \bibinfo {year} {2018})\
  pp.\ \bibinfo {pages} {500--532}\BibitemShut {NoStop}%
\bibitem [{\citenamefont {Cohen}\ \emph {et~al.}(2020)\citenamefont {Cohen},
  \citenamefont {Garay},\ and\ \citenamefont {Zikas}}]{cohen2020broadcast}%
  \BibitemOpen
  \bibfield  {author} {\bibinfo {author} {\bibfnamefont {R.}~\bibnamefont
  {Cohen}}, \bibinfo {author} {\bibfnamefont {J.}~\bibnamefont {Garay}}, \ and\
  \bibinfo {author} {\bibfnamefont {V.}~\bibnamefont {Zikas}},\ }in\ \href@noop
  {} {\emph {\bibinfo {booktitle} {Annual International Conference on the
  Theory and Applications of Cryptographic Techniques}}}\ (\bibinfo
  {organization} {Springer},\ \bibinfo {year} {2020})\ pp.\ \bibinfo {pages}
  {828--858}\BibitemShut {NoStop}%
\bibitem [{\citenamefont {Badrinarayanan}\ \emph {et~al.}(2018)\citenamefont
  {Badrinarayanan}, \citenamefont {Goyal}, \citenamefont {Jain}, \citenamefont
  {Kalai}, \citenamefont {Khurana},\ and\ \citenamefont
  {Sahai}}]{badrinarayanan2018promise}%
  \BibitemOpen
  \bibfield  {author} {\bibinfo {author} {\bibfnamefont {S.}~\bibnamefont
  {Badrinarayanan}}, \bibinfo {author} {\bibfnamefont {V.}~\bibnamefont
  {Goyal}}, \bibinfo {author} {\bibfnamefont {A.}~\bibnamefont {Jain}},
  \bibinfo {author} {\bibfnamefont {Y.~T.}\ \bibnamefont {Kalai}}, \bibinfo
  {author} {\bibfnamefont {D.}~\bibnamefont {Khurana}}, \ and\ \bibinfo
  {author} {\bibfnamefont {A.}~\bibnamefont {Sahai}},\ }in\ \href@noop {}
  {\emph {\bibinfo {booktitle} {Annual International Cryptology Conference}}}\
  (\bibinfo {organization} {Springer},\ \bibinfo {year} {2018})\ pp.\ \bibinfo
  {pages} {459--487}\BibitemShut {NoStop}%
\bibitem [{\citenamefont {Halevi}\ \emph {et~al.}(2018)\citenamefont {Halevi},
  \citenamefont {Hazay}, \citenamefont {Polychroniadou},\ and\ \citenamefont
  {Venkitasubramaniam}}]{halevi2018round}%
  \BibitemOpen
  \bibfield  {author} {\bibinfo {author} {\bibfnamefont {S.}~\bibnamefont
  {Halevi}}, \bibinfo {author} {\bibfnamefont {C.}~\bibnamefont {Hazay}},
  \bibinfo {author} {\bibfnamefont {A.}~\bibnamefont {Polychroniadou}}, \ and\
  \bibinfo {author} {\bibfnamefont {M.}~\bibnamefont {Venkitasubramaniam}},\
  }in\ \href@noop {} {\emph {\bibinfo {booktitle} {Annual International
  Cryptology Conference}}}\ (\bibinfo {organization} {Springer},\ \bibinfo
  {year} {2018})\ pp.\ \bibinfo {pages} {488--520}\BibitemShut {NoStop}%
\bibitem [{\citenamefont {Applebaum}\ \emph {et~al.}(2020)\citenamefont
  {Applebaum}, \citenamefont {Kachlon},\ and\ \citenamefont
  {Patra}}]{applebaum2020round}%
  \BibitemOpen
  \bibfield  {author} {\bibinfo {author} {\bibfnamefont {B.}~\bibnamefont
  {Applebaum}}, \bibinfo {author} {\bibfnamefont {E.}~\bibnamefont {Kachlon}},
  \ and\ \bibinfo {author} {\bibfnamefont {A.}~\bibnamefont {Patra}},\ }in\
  \href@noop {} {\emph {\bibinfo {booktitle} {Electronic Colloquium on
  Computational Complexity (ECCC)}}},\ Vol.~\bibinfo {volume} {27}\ (\bibinfo
  {year} {2020})\ p.~\bibinfo {pages} {76}\BibitemShut {NoStop}%
\bibitem [{\citenamefont {Cr{\'e}peau}\ \emph {et~al.}(2002)\citenamefont
  {Cr{\'e}peau}, \citenamefont {Gottesman},\ and\ \citenamefont
  {Smith}}]{crepeau2002secure}%
  \BibitemOpen
  \bibfield  {author} {\bibinfo {author} {\bibfnamefont {C.}~\bibnamefont
  {Cr{\'e}peau}}, \bibinfo {author} {\bibfnamefont {D.}~\bibnamefont
  {Gottesman}}, \ and\ \bibinfo {author} {\bibfnamefont {A.}~\bibnamefont
  {Smith}},\ }in\ \href@noop {} {\emph {\bibinfo {booktitle} {Proceedings of
  the thiry-fourth annual ACM symposium on Theory of computing}}}\ (\bibinfo
  {year} {2002})\ pp.\ \bibinfo {pages} {643--652}\BibitemShut {NoStop}%
\bibitem [{\citenamefont {Dupuis}\ \emph {et~al.}(2010)\citenamefont {Dupuis},
  \citenamefont {Nielsen},\ and\ \citenamefont {Salvail}}]{dupuis2010secure}%
  \BibitemOpen
  \bibfield  {author} {\bibinfo {author} {\bibfnamefont {F.}~\bibnamefont
  {Dupuis}}, \bibinfo {author} {\bibfnamefont {J.~B.}\ \bibnamefont {Nielsen}},
  \ and\ \bibinfo {author} {\bibfnamefont {L.}~\bibnamefont {Salvail}},\ }in\
  \href@noop {} {\emph {\bibinfo {booktitle} {Annual Cryptology Conference}}}\
  (\bibinfo {organization} {Springer},\ \bibinfo {year} {2010})\ pp.\ \bibinfo
  {pages} {685--706}\BibitemShut {NoStop}%
\bibitem [{\citenamefont {Dulek}\ \emph {et~al.}(2020)\citenamefont {Dulek},
  \citenamefont {Grilo}, \citenamefont {Jeffery}, \citenamefont {Majenz},\ and\
  \citenamefont {Schaffner}}]{dulek2020secure}%
  \BibitemOpen
  \bibfield  {author} {\bibinfo {author} {\bibfnamefont {Y.}~\bibnamefont
  {Dulek}}, \bibinfo {author} {\bibfnamefont {A.~B.}\ \bibnamefont {Grilo}},
  \bibinfo {author} {\bibfnamefont {S.}~\bibnamefont {Jeffery}}, \bibinfo
  {author} {\bibfnamefont {C.}~\bibnamefont {Majenz}}, \ and\ \bibinfo {author}
  {\bibfnamefont {C.}~\bibnamefont {Schaffner}},\ }in\ \href@noop {} {\emph
  {\bibinfo {booktitle} {Annual International Conference on the Theory and
  Applications of Cryptographic Techniques}}}\ (\bibinfo {organization}
  {Springer},\ \bibinfo {year} {2020})\ pp.\ \bibinfo {pages}
  {729--758}\BibitemShut {NoStop}%
\bibitem [{\citenamefont {Van{\'\i}{\v{c}}ek}\ and\ \citenamefont
  {Heller}(2003)}]{vanivcek2003semiclassical}%
  \BibitemOpen
  \bibfield  {author} {\bibinfo {author} {\bibfnamefont {J.}~\bibnamefont
  {Van{\'\i}{\v{c}}ek}}\ and\ \bibinfo {author} {\bibfnamefont {E.~J.}\
  \bibnamefont {Heller}},\ }\href@noop {} {\bibfield  {journal} {\bibinfo
  {journal} {Physical Review E}\ }\textbf {\bibinfo {volume} {68}},\ \bibinfo
  {pages} {056208} (\bibinfo {year} {2003})}\BibitemShut {NoStop}%
\bibitem [{\citenamefont {Biamonte}\ \emph {et~al.}(2017)\citenamefont
  {Biamonte}, \citenamefont {Wittek}, \citenamefont {Pancotti}, \citenamefont
  {Rebentrost}, \citenamefont {Wiebe},\ and\ \citenamefont
  {Lloyd}}]{biamonte2017quantum}%
  \BibitemOpen
  \bibfield  {author} {\bibinfo {author} {\bibfnamefont {J.}~\bibnamefont
  {Biamonte}}, \bibinfo {author} {\bibfnamefont {P.}~\bibnamefont {Wittek}},
  \bibinfo {author} {\bibfnamefont {N.}~\bibnamefont {Pancotti}}, \bibinfo
  {author} {\bibfnamefont {P.}~\bibnamefont {Rebentrost}}, \bibinfo {author}
  {\bibfnamefont {N.}~\bibnamefont {Wiebe}}, \ and\ \bibinfo {author}
  {\bibfnamefont {S.}~\bibnamefont {Lloyd}},\ }\href@noop {} {\bibfield
  {journal} {\bibinfo  {journal} {Nature}\ }\textbf {\bibinfo {volume} {549}},\
  \bibinfo {pages} {195} (\bibinfo {year} {2017})}\BibitemShut {NoStop}%
\bibitem [{\citenamefont {Vaccaro}\ \emph {et~al.}(2007)\citenamefont
  {Vaccaro}, \citenamefont {Spring},\ and\ \citenamefont
  {Chefles}}]{vaccaro2007quantum}%
  \BibitemOpen
  \bibfield  {author} {\bibinfo {author} {\bibfnamefont {J.~A.}\ \bibnamefont
  {Vaccaro}}, \bibinfo {author} {\bibfnamefont {J.}~\bibnamefont {Spring}}, \
  and\ \bibinfo {author} {\bibfnamefont {A.}~\bibnamefont {Chefles}},\
  }\href@noop {} {\bibfield  {journal} {\bibinfo  {journal} {Physical Review
  A}\ }\textbf {\bibinfo {volume} {75}},\ \bibinfo {pages} {012333} (\bibinfo
  {year} {2007})}\BibitemShut {NoStop}%
\bibitem [{\citenamefont {Brakerski}\ and\ \citenamefont
  {Yuen}(2020)}]{brakerski2020quantum}%
  \BibitemOpen
  \bibfield  {author} {\bibinfo {author} {\bibfnamefont {Z.}~\bibnamefont
  {Brakerski}}\ and\ \bibinfo {author} {\bibfnamefont {H.}~\bibnamefont
  {Yuen}},\ }\href@noop {} {\bibfield  {journal} {\bibinfo  {journal} {arXiv
  preprint arXiv:2006.01085}\ } (\bibinfo {year} {2020})}\BibitemShut {NoStop}%
\bibitem [{\citenamefont {Ambainis}\ \emph {et~al.}(2000)\citenamefont
  {Ambainis}, \citenamefont {Mosca}, \citenamefont {Tapp},\ and\ \citenamefont
  {De~Wolf}}]{ambainis2000private}%
  \BibitemOpen
  \bibfield  {author} {\bibinfo {author} {\bibfnamefont {A.}~\bibnamefont
  {Ambainis}}, \bibinfo {author} {\bibfnamefont {M.}~\bibnamefont {Mosca}},
  \bibinfo {author} {\bibfnamefont {A.}~\bibnamefont {Tapp}}, \ and\ \bibinfo
  {author} {\bibfnamefont {R.}~\bibnamefont {De~Wolf}},\ }in\ \href@noop {}
  {\emph {\bibinfo {booktitle} {Proceedings 41st Annual Symposium on
  Foundations of Computer Science}}}\ (\bibinfo {organization} {IEEE},\
  \bibinfo {year} {2000})\ pp.\ \bibinfo {pages} {547--553}\BibitemShut
  {NoStop}%
\bibitem [{\citenamefont {Agarwal}\ \emph {et~al.}(2020)\citenamefont
  {Agarwal}, \citenamefont {Bartusek}, \citenamefont {Goyal}, \citenamefont
  {Khurana},\ and\ \citenamefont {Malavolta}}]{agarwal2020post}%
  \BibitemOpen
  \bibfield  {author} {\bibinfo {author} {\bibfnamefont {A.}~\bibnamefont
  {Agarwal}}, \bibinfo {author} {\bibfnamefont {J.}~\bibnamefont {Bartusek}},
  \bibinfo {author} {\bibfnamefont {V.}~\bibnamefont {Goyal}}, \bibinfo
  {author} {\bibfnamefont {D.}~\bibnamefont {Khurana}}, \ and\ \bibinfo
  {author} {\bibfnamefont {G.}~\bibnamefont {Malavolta}},\ }\href@noop {}
  {\bibfield  {journal} {\bibinfo  {journal} {arXiv preprint arXiv:2005.12904}\
  } (\bibinfo {year} {2020})}\BibitemShut {NoStop}%
\bibitem [{\citenamefont {Bartusek}\ \emph {et~al.}(2020)\citenamefont
  {Bartusek}, \citenamefont {Coladangelo}, \citenamefont {Khurana},\ and\
  \citenamefont {Ma}}]{bartusek2020round}%
  \BibitemOpen
  \bibfield  {author} {\bibinfo {author} {\bibfnamefont {J.}~\bibnamefont
  {Bartusek}}, \bibinfo {author} {\bibfnamefont {A.}~\bibnamefont
  {Coladangelo}}, \bibinfo {author} {\bibfnamefont {D.}~\bibnamefont
  {Khurana}}, \ and\ \bibinfo {author} {\bibfnamefont {F.}~\bibnamefont {Ma}},\
  }\href@noop {} {\enquote {\bibinfo {title} {On the round complexity of
  two-party quantum computation},}\ } (\bibinfo {year} {2020}),\ \Eprint
  {http://arxiv.org/abs/2011.11212} {arXiv:2011.11212 [quant-ph]} \BibitemShut
  {NoStop}%
\bibitem [{\citenamefont {Alon}\ \emph {et~al.}(2020)\citenamefont {Alon},
  \citenamefont {Chung}, \citenamefont {Chung}, \citenamefont {Huang},
  \citenamefont {Lee},\ and\ \citenamefont {Shen}}]{cryptoeprint:2020:1464}%
  \BibitemOpen
  \bibfield  {author} {\bibinfo {author} {\bibfnamefont {B.}~\bibnamefont
  {Alon}}, \bibinfo {author} {\bibfnamefont {H.}~\bibnamefont {Chung}},
  \bibinfo {author} {\bibfnamefont {K.-M.}\ \bibnamefont {Chung}}, \bibinfo
  {author} {\bibfnamefont {M.-Y.}\ \bibnamefont {Huang}}, \bibinfo {author}
  {\bibfnamefont {Y.}~\bibnamefont {Lee}}, \ and\ \bibinfo {author}
  {\bibfnamefont {Y.-C.}\ \bibnamefont {Shen}},\ }\href@noop {} {\enquote
  {\bibinfo {title} {Round efficient secure multiparty quantum computation with
  identifiable abort},}\ }\bibinfo {howpublished} {Cryptology ePrint Archive,
  Report 2020/1464} (\bibinfo {year} {2020}),\ \bibinfo {note}
  {\url{https://eprint.iacr.org/2020/1464}}\BibitemShut {NoStop}%
\bibitem [{\citenamefont {Tate}\ and\ \citenamefont
  {Xu}(2003)}]{tate2003garbled}%
  \BibitemOpen
  \bibfield  {author} {\bibinfo {author} {\bibfnamefont {S.~R.}\ \bibnamefont
  {Tate}}\ and\ \bibinfo {author} {\bibfnamefont {K.}~\bibnamefont {Xu}},\
  }\href@noop {} {\bibfield  {journal} {\bibinfo  {journal} {CoPS Lab,
  University of North Texas, Tech. Rep}\ }\textbf {\bibinfo {volume} {2}},\
  \bibinfo {pages} {2003} (\bibinfo {year} {2003})}\BibitemShut {NoStop}%
\end{thebibliography}%


\renewcommand{\theequation}{S\arabic{equation}}
\setcounter{equation}{0}
\renewcommand{\thefigure}{S\arabic{figure}}
\setcounter{figure}{0}
\renewcommand{\thesection}{S\arabic{section}}
\setcounter{section}{0}

\onecolumngrid


\begin{center}
	{\bf \Large 
		Supplementary Materials to\vspace*{0.3cm}\\ 
		\emph{Constant-round Multi-party Quantum Computation for Constant Parties}
	}
\end{center}

\vspace*{0.3cm}
{\center{
		\hspace*{0.1\columnwidth}\begin{minipage}[c]{0.8\columnwidth}
			In these Supplementary Materials, 
			(i)  we provide a security proof for QOTP; 
			(ii) we provide a security proof of MPQC for two parties;
			(iii) and we provide a security proof of MPQC for multiple parties.
		\end{minipage}
	}
}

\section{Security of QOTP}
\label{Appsec:QOTP}
In this section, we show the security of QOTP.

In the classical bit case, an adversary is said to be unable to
determine the value of a bit if he cannot be certain whether the probability
of 0 is strictly greater or less than the probability of 1.

In the quantum case, for a qubit, we similarly define that,
when enumerating the plane that contains the origin in the Bloch sphere space,
if the maximal probability difference that the adversary can determine the actual 
qubit is on one side of the plane than the other side is 0, then the qubit is 
information-theoretically secure.

Under this definition, sending four copies of Alice's qubit encrypted under QOTP 
is secure, as even infinite copies of these values are sent, Eve still cannot 
distinguish Alice's qubit by the four points on the Bloch sphere, which are related by $I$, $\sigma_x$, $\sigma_y$, $\sigma_z$.
This is because every plane passing through the origin is unable to separate the probability space of this qubit (the four points)
into two unequal probability regions.

\section{Proof of the two party case}
\label{Appsec:TwoParty}
\begin{definition}[Advantage for a distinguisher $D$] 
Given a classical or quantum string $r$ from the real world $\textsf{Real}$ or the simulated world $\textsf{Sim}$,
the distinguisher $D$ outputs either 0 or 1.
The advantage of $D$ is defined as 
\begin{equation}
Adv[D] := | Pr(D(\textsf{Real})=1) - Pr(D(\textsf{Sim})=1) |.
\end{equation}
\end{definition}
 
 Alice's simulator $\textsf{Sim}$ is as follows. The simulator has an input $F(x_1,x_2,\cdots, x_n)$ and a topology of the circuit  $\Gamma$. 
It first generates a garbled circuit $\hat{E} (F(x_1,x_2,\cdots, x_n), \Gamma)$ with inputs  $F(x_1,x_2,\cdots, x_n)$ and $\Gamma$.
 It then takes the value that corresponds to Bob's label as a replacement of Bob's $y_b$.
This is given as part of the input to the simulator $S_1$ in OT. Alice's simulator then generates whatever $S_1$ outputs. 
 
To show this simulator works, we use a hybrid argument. 
View Alice's simulator as Game 2, the real world as Game 0, and define Game 1 as follows:
 
 \emph{Bob generates a DQRE as usual and sends Alice the part that he can compute. Next instead of a usual OT protocol between 
 Alice and Bob, the simulator $S_1$ in OT takes the place of Bob to communicate with Alice.}
 
 We first prove two lemmas. 
 
 \begin{lemma}
 If an adversary can distinguish between Game 1 and the real world $\textsf{Real}$ with an advantage larger than $\epsilon$, then it can break the $\epsilon$-secure OT. Consequently, if the OT is $\epsilon$-secure, then no adversary can distinguish between Game 1 and  the real world $\textsf{Real}$ with an advantage larger than $\epsilon$.
 \end{lemma}
 
 \begin{proof}
 Since the first step of Game 1 and the real world are the same (Bob generates DQRE and sends Alice the parts except $y_0$ and $y_1$),
 the distribution of $(y_0, y_1, b)$ are the same for these two worlds. If a distinguisher $A$ can distinguish these two worlds with advantage larger than $\epsilon$, then at least for one tuple $(y_0', y_1', b')$, the distinguisher $A$ would have an advantage larger than $\epsilon$. Now consider a distinguisher $B$ of OT that takes  this tuple $(y_0', y_1', b')$ as his input, and performs the following attack:
 \begin{enumerate}
 \item $B$ conditions on that Alice and Bob's joint view is $(y_0', y_1', b')$ after Bob generates the DQRE.
 \item Alice interacts with $B$'s challenger.
 \item Finally, $B$ outputs whatever $A$ outputs.
 \end{enumerate}
 Now when $B$ is in Experiment 0 of its attack game, it perfectly mimics the behaviour of $A$ in the real world $\textsf{Real}$, and when $B$ is in Experiment 1 of its attack game, it perfectly mimics the behaviour of $A$ in Game 1.  Hence, $Pr(A(\textsf{Real})=1)$ is the same as $Pr(B(\textrm{Experiment 0})=1)$, and $Pr(A(\textrm{Game 1})=1)$ is the same as $Pr(B(\textrm{Experiment 1})=1)$.
Therefore, the advantage of $B$ is precisely the same as the advantage of $A$, which is larger than $\epsilon$. Hence this breaks the property of a $\epsilon$-secure OT.
 \end{proof}
 
  \begin{lemma}
If the DQRE is $\epsilon$-private, then no adversary can distinguish Game 1 with the simulated world $\textsf{Sim}$ with an advantage larger than $\epsilon$.
 \end{lemma}
 
 \begin{proof}
Let $A$ be the distinguisher for Game 1 and the simulated world $\textsf{Sim}$.
Given an adversary $B$ of DQRE, it performs the following attack.
 \begin{enumerate}
 \item It first takes the DQRE from the challenger, and sends the parts except the one that corresponds to Alice's input (denoted by $y_b$) to Alice.
 \item It uses the simulator $S_1$ of OT that uses $y_b$ and $b$ to generate a transcript between Bob and Alice.
 \item Finally, $B$ outputs whatever $A$ outputs.
 \end{enumerate}
  Now when $B$ is in Experiment 0 of its attack game, it perfectly mimics the behaviour of $A$ in Game 1, and when $B$ is in Experiment 1 of its attack game, it perfectly mimics the behaviour of $A$ in the simulated world $\textsf{Sim}$.  Hence, $Pr(A(\textsf{Sim})=1)$ is the same as $Pr(B(\textrm{Experiment 1})=1)$, and $Pr(A(\textrm{Game 1})=1)$ is the same as $Pr(B(\textrm{Experiment 0})=1)$.
Therefore, the advantage of $B$ is precisely the same as the advantage of $A$, therefore if the DQRE is $\epsilon$-private, then no adversary can distinguish Game 1 from the simulated world $\textsf{Sim}$ with an advantage larger than $\epsilon$.
 \end{proof}

 Now we return to the security proof for the original problem. By the lemmas, we have 
 \begin{eqnarray}
Adv[D]& =& | Pr(D(\textsf{Real})=1) - Pr(D(\textsf{Sim})=1)  | \nonumber \\
&\le & | Pr(D(\textsf{Real})=1) - Pr(D(\textrm{Game 1})=1)  | \nonumber \\
&& +  | Pr(D(\textrm{Game 1})=1) - Pr(D(\textsf{Sim})=1) | \nonumber \\
&\le & \epsilon + \epsilon \\
& =& 2\epsilon, \nonumber 
\end{eqnarray}
 hence, the protocol is secure.

\section{Security proof for the multi-party case}
\label{AppSec:MultiProof}
In this section, we show the security of MPQC for multiple parties. 
Intuitively, the proof can be inferred from two facts.
Firstly, after the DQRE is generated, according to its property,  nothing can be inferred except what can be inferred from the output of DQRE
$F(\rho_1, \cdots, \rho_n)$. Secondly, the label and the semantic value of each wire is decoupled, which can only be known if all parties reveal their shares of qubit flipping. 

Now, we are ready to present the formal proof.
Assume the adversary Eve controls $n-1$ parties. By the symmetry of the protocol, without loss of generality, we assume Eve controls $P_2, \cdots, P_n$ and aims to learn the private input of $P_1$.

The proof consists of two steps. For the first step, we note that the active wire labels and gate labels are essentially random except that they can be used to compute $F(\rho_1, \cdots, \rho_n)$, as followed from the property of DQRE. In particular, the active wire labels by themselves give no information on $P_1$'s input except which can be inferred from $F(\rho_1, \cdots, \rho_n)$. 

In the second step, we note that the active wire label combined with the $n-1$ shares that Eve possess for this wire is independent of the semantic value of this wire. In particular, the semantic values of the input wires of $P_1$ are hided from Eve.

Three remarks are in order. First, the correspondence between a wire's active label and its semantic value is known for those wires that can be computed based solely on $P_2,\cdots, P_n$'s inputs. But once $P_1$'s input is involved, the correspondence becomes completely opaque to the adversary. Secondly, it is instructive to see why a garbled circuit generated by a single party would fail in the security proof. In this case, the second step in the security proof no longer holds, as the active wire label completely reveals the semantic value for any wire once Eve controls this single party that generates the garbled circuit and the party that obtains the active wire labels. Thirdly, it is also instructive to see why a failure of the second step does not hurt the security for the two party case ($n=2$). This is because the active labels and the correspondence between wire labels and wire values are held by two different parties. Eve cannot obtain both these information, as she can only control $n-1=1$ party.

\end{document}